\documentclass[10pt]{article}
\usepackage{amssymb}
\usepackage{amsfonts}
\usepackage{amsthm}
\usepackage{amsmath}
\usepackage{bm}	
\usepackage{bbm}
\usepackage{upgreek}
\usepackage{hyperref}
\usepackage{comment}
\usepackage{xspace}
\usepackage{chngcntr}
\usepackage{graphicx,accents}
\usepackage{tikz}

\title{Critical exponent for the magnetization of the weakly coupled $\phi_4^4 $ model } 
\author{Martin Lohmann\thanks{Department of Mathematics, University of British Columbia, 1984 Mathematics Road, V6T 1Z2 Vancouver, Canada. email: marlohmann@math.ubc.ca}}
\date{}
\allowdisplaybreaks

\newcommand{\mquad}{\!\!\!\!\!\!\!\!\!\!\!\!\!\!\!\!\!\!\!\!\!\!\!}
\newcommand{\qqquad}{\qquad\qquad\qquad\qquad\qquad}
\newcommand{\erem}{ \begin{flushright} \(\diamond\)  \end{flushright}}
\newcommand{\nnb}{\nonumber \\}
\newcommand{\bd}[1]{\mathbf{#1}}
\newcommand{\mc}[1]{\mathcal{#1}}
\newcommand{\mf}[1]{\mathfrak{#1}}
\newcommand{\mb}[1]{\mathbb{#1}}
\renewcommand{\sf}[1]{\mathsf{#1}}
\newcommand{\ul}[1]{\underline{#1}}

\newcommand{\ud}{\,\mathrm{d}}

\newcommand{\sums}[1]{\sum_{\substack{#1}}}
\newcommand{\half}[0]{{\frac{1}{2}}}

\newcommand{\const}{\,\text{const}\,}
\newcommand{\eps}{\varepsilon}
\newcommand{\la}{\langle}
\newcommand{\ra}{\rangle}
\newcommand{\sfx}{\sf x}
\newcommand{\sfy}{\sf y}
\newcommand{\sfz}{\sf z}

\newcommand{\sfg}{\sf g}
\newcommand{\sfX}{\sf X}

\newcommand{\sfB}{\sf B}
\newcommand{\sfA}{\sf A}

\newcommand{\ttl}{\ell}
\newcommand{\gk}{\cite{gk1985}\xspace}

\makeatletter
\DeclareRobustCommand{\ivec}[1]{%
	\mathpalette\do@ivec{#1}%
}
\newcommand{\do@ivec}[2]{%
	\fix@ivec{#1}{+}%
	\reflectbox{$\m@th#1\vec{\reflectbox{$\fix@ivec{#1}{-}\m@th#1#2\fix@ivec{#1}{+}$}}$}%
	\fix@ivec{#1}{-}%
}
\newcommand{\fix@ivec}[2]{%
	\ifx#1\displaystyle
	\mkern#23mu
	\else
	\ifx#1\textstyle
	\mkern#23mu
	\else
	\ifx#1\scriptstyle
	\mkern#22mu
	\else
	\mkern#22mu
	\fi
	\fi
	\fi
}

\makeatother

\theoremstyle{plain}
\newtheorem{thm}{Theorem}[section]
\newtheorem{lem}[thm]{Lemma}
\newtheorem{prop}[thm]{Proposition}

\theoremstyle{definition}

\theoremstyle{remark}
\newtheorem{rem}[thm]{Remark}

\numberwithin{equation}{section}

\newcounter{parag}
\counterwithin{parag}{section}
\newcommand{\para}[2]{\refstepcounter{parag}\label{#2}\noindent\textbf{\arabic{section}.\arabic{parag} #1}}

\newcommand{\blabel}[1]{\stepcounter{equation}\label{#1}\tag{\textbf{\arabic{section}.\arabic{equation}}}}

\begin{document}

\maketitle

\begin{abstract}
We consider the weakly coupled $\phi^4 $ theory on $\mb Z^4 $, in a weak magnetic field $h$, and at the chemical potential $\nu_c $ for which the theory is critical if $h=0$. We prove that, as $h\searrow0$, the magnetization of the model behaves as $(h\log h^{-1})^{\frac 13} $, and so exhibits a logarithmic correction to mean field scaling behavior. This result is well known to physicists, but had never been proven rigorously. Our proof uses the classic construction of the critical theory by Gawedzki and Kupiainen \gk, and a cluster expansion with large blocks.
\end{abstract}

\section{Introduction}

In this paper, we study the magnetization of the weakly coupled $\phi^4$ theory in dimension $4$, which is defined by
\begin{align}\label{eqmagn}
m(h) &= \lim_{\substack{\Lambda_0\subset \mb Z^4 \\ \Lambda_0\to \mb Z^4}}\Bigg[\tfrac1{|\Lambda_0|} \partial_\eps \log \int e^{- \sum_{\xi\in \Lambda_0}\frac 12 \psi_\xi(-\Delta\psi)_\xi + \frac{\nu_c} 2\psi_\xi^2 + \frac{g_0}4\psi_\xi^4 - (h+\eps)\psi_\xi }\prod_{\xi\in\Lambda_0}\ud\psi_\xi\ \bigg|_{\eps=0}\Bigg]
\end{align}
For $g_0>0 $ small, if the limit is taken over large tori $\Lambda_0 = \mb Z^4 / L^N\mb Z^4 $ with $L\in\mb N $ fixed large, $N\to\infty $, and with $\Delta $ the Laplacian with periodic boundary conditions, the measure of the above integral was first analyzed at $h=0 $ in the famous papers \cite{gk1985,feldman1987}. It was proven there that there is a value $\nu_c = \nu_c(g_0)<0 $ of the chemical potential for which the model is critical, that is, displays long range correlations. The core of their proof is a rigorous implementation of the renormalization group (RG), and uses the asymptotic freedom of the RG flow. \\
It is of special interest in physics to understand the behavior of thermodynamic quantities as such critical points are approached. In dimensions $d>4$, it can be proven that this behavior follows the predictions of mean field theory, see e.g. \cite{aizenman1982,fernandez2013random}. In $d=4$, when approaching the critical point from the high temperature phase, because of asymptotic freedom, one expects logarithmic corrections to the mean field behavior, but proofs are much harder than in $d>4$. For the magnetization, with $\nu_c $ as constructed in \cite{gk1985,feldman1987}, physicists have conjectured 
\begin{align}\label{logcorr}
m(h) \sim \Big(\frac{3h\log h^{-1}} {16\pi^2}\Big)^{\frac13}\qquad \text{as} \qquad h\searrow 0,
\end{align}
as opposed to the mean field prediction $m(h)\sim h^{\frac13} $.\\
For the magnetic susceptibility and the correlation length, a proof\footnote{Earlier rigorous results that indicate logarithmic corrections were obtained by more elegant methods in \cite{1983NuPhB.225..261A}} of logarithmic corrections was first given in \cite{hara1987rigorous},\cite{Hara1987}. They study the off-critical flow first by comparison with the critical construction of \gk, until the two flows differ sufficiently to view the effective model as the perturbation of a massive Gaussian model. They then adapt the arguments of \gk to massive RG flows to finish the construction. More recently, this result was extended to the multi-component $\phi^4 $ model \cite{bauerschmidt2014scaling} and to the weakly self avoiding random walk \cite{bauerschmidt2015logarithmic}. \\
All the results discussed above concern models that are symmetric under $\phi\to-\phi $, while the introduction of the magnetic field $h$ in (\ref{eqmagn}) breaks this symmetry. From the RG point of view, the broken symmetry introduces two additional relevant coupling constants whose flow would have to be controlled, and even though the analysis could certainly proceed along the same general lines as in the symmetric case, this adds an additional layer of technical complexity to an already complicated argument, and it is probably for this reason that logarithmic corrections to the magnetization, despite featuring prominently in physics textbook discussions of phase transitions, have never been proven mathematically. The only rigorous result on this question, the bound $c\ h^{\frac13}\leq m(h)\leq C\ h^{\frac13}\log h^{-1} $, was established in \cite{1986JSP....44..393A} for the four dimensional Ising model without using the RG.\\
In this paper, we will show that (\ref{logcorr}) can be obtained without running into the technical difficulties just mentioned. Indeed, we will argue that the \textit{critical} construction of \gk is already powerful enough to generate an effective representation for $m(h)$ that can be analyzed without any RG flow, but rather by a single cluster expansion with large blocks. This cluster expansion has a slightly different stability analysis than, and features certain scaling arguments similar to, a single RG step (in which a cluster expansion with unit size blocks is involved). \\
We will prove the following theorem

\begin{thm}\label{mainthm}
Choose any $\delta>0 $ small, and consider $m(h)$ as in (\ref{eqmagn}), with the limit taken over tori $\Lambda_0= \mb Z^4/L^N\mb Z^4 $, $L\in\mb N_{\text{odd}} $ fixed large, $N\to\infty$, with $\Delta $ the Laplacian
\begin{align}
\sum_{\xi\in\Lambda_0} \psi_\xi(-\Delta\psi)_\xi &= \sums{\xi,\xi'\in\Lambda_0\\ |\xi-\xi'|_1=1}(\psi_\xi-\psi_{\xi'})^2 + L^{-4N} \Big(\sum_{\xi\in\Lambda_0}\psi_\xi\Big)^2
\end{align}
Let $g_0>0 $ be small enough, depending on $\delta$, and set $\nu_c $ equal to the critical chemical potential constructed in \gk. Then, if $h>0$ is small enough (depending on $g_0$),
\begin{align}\label{logcorr2}
m(h) = \Big(\frac{3h\log h^{-1}} {16\pi^2}\Big)^{\frac13}(1+\mc E(h)), 
\end{align}
where $|\mc E(h)|\leq \delta $.
\end{thm}

\begin{rem}\label{remstrongsympt}
Ideally, one would like to prove $\mc E(h) = o(1) $ as $h\searrow 0$. As we shall see, this would follow immediately from our argument if the results of \gk were formulated in a slightly stronger form. In fact, it is believable that a small improvement on \gk would give $\mc E(h) = O(1/\log\log h^{-1}) $, see Remark \ref{remgkimprovement}. Since our focus is on showing how the critical construction of \gk can be used as a black box to study critical exponents, we will not pursue any modifications of their argument, however.\erem
\end{rem}

\begin{rem}
The construction of \gk is particularly well suited to be used as the input of a single cluster expansion. While the more recent implementation of the renormalization group by Brydges and collaborators (\cite{brydges1990,Brydges1995,Brydges2015,2016arXiv160609541A}) has been developed in more detail and also been applied to critical phenomena of other interesting models (such as long range models below the critical dimension \cite{Brydges2003,Slade2018,Lohmann2017}, or the dipole gas \cite{Dimock2000,2013arXiv1311.2237F}), and while this method shares the intuitive inductive nature of the \gk argument and uses similar expansions, there are obstacles to using it as the basis of the current paper. In particular, the behavior of effective interactions at large fields is not tracked as precisely as in \gk, and currently our method relies on such precise large field bounds. \erem
\end{rem}

\noindent
The plan of the paper is as follows. In section \ref{sgk}, we recall all results of \gk which will be used as an input for our argument, in particular, the effective representation for $m(h) $ that follows from their argument. A rather detailed heuristic derivation of (\ref{logcorr2}) based on this representation is given in section \ref{remheur}, which is the conceptual core of this paper, and serves as a guide for the following sections. In section \ref{sce}, we set up a cluster expansion for the effective representation. Estimates on the expansion are proven in section \ref{bounds}, and they are used in section \ref{pmt} to derive Theorem \ref{mainthm}. Some technical lemmas are deferred to the appendix.

\begin{rem}\label{remsmallnot}\textbf{Smallness conditions.} This paper contains a number of conditions on various small and large parameters. These conditions will be emphasized by boldface equation numbers. We will implicitly assume any such condition in all arguments that follow after they have been stated. Clearly, it is important that all conditions can be met by choosing the basic parameters of the model. In the present remark, we give a summary on the order in which parameters are chosen, and the reader is invited to refer to it whenever necessary. \\
The basic small parameter in Theorem \ref{mainthm} is $\delta>0 $. There is also a condition $L\geq L_0 $, with some absolute constant $L_0$, for which the \gk construction works. We shall fix these two parameters first, and conditions on all other parameters will depend on them.\\
Our proof will be based on a cluster expansion, involving a scale parameter $\mc L $, and relies on the polymer gas representation of \gk, with scale parameter $\ell$. The basic quantity that has to be chosen large enough, depending on $\delta$, will be $\frac{\mc L}\ell $. Depending on $\frac{\mc L}\ell $, $\ell$ itself has to be chosen large enough. (In other words, $x =\frac{\mc L}\ell $ is chosen large depending on $\delta$, $\ell$ is chosen large depending on $x$, and we then set $\mc L = x\cdot\ell $).\\
After these choices have been made, there will be conditions $c_1\geq c_1(\ell) $ and $c_2\geq c_2(c_1) $ on two constants in technical bounds, and finally conditions $0<g_0<g_0(c_2) $ and $0<h<h(g_0) $ on the coupling constant and the magnetic field. Depending on $h$ and $ \mc L  $, we will define a carefully chosen ``magnetization scale'' $n$, which is the scale at which we terminate the critical flow of \gk and perform our cluster expansion. We will also choose a (large) radius $\mf r$ of analyticity in the field $\psi$ and a (small) radius $\upepsilon $ of analyticity in $\eps$, that arise in the bounds we prove on the cluster expansion. \\ [5pt]
\textbf{Thermodynamic limit.} Since the thermodynamic limit $N\to\infty $ in (\ref{eqmagn}) is performed before we consider $h\to 0 $, we will always work with $N\geq N(h)$ that is much larger than the magnetization scale (plus the scale of the cluster expansion). Our analysis will be uniform in $N$, as long as $N\geq N(h)$. The limit $N\to\infty $ will be discussed only shortly at the end of section \ref{pmt}. From now until then, we shall regard $N\geq N(h) $ as fixed.\\
The derivative $\partial_\eps$ in (\ref{eqmagn}) is performed and evaluated at $\eps=0 $ before the limit $N\to\infty $. We could therefore always think of $\eps$ as being as small as we like, even depending on $N$, although we shall often be interested in analyticity properties of certain functions of $\eps$, in which case we consider any $|\eps|\leq\upepsilon $, with $\upepsilon$ the radius of analyticity mentioned above, which is chosen independent of $N$. \\[5pt]
\textbf{Notation.} We shall later make use of the following two notations: First, we write $a\ll b $ to mean that $\frac ba $ can be made as big as we want, as long as the basic parameter $\frac{\mc L}\ell $ is chosen large enough. For example, $\ell\gg1 $. Second, for some given number $a\in\mb R $, we shall write $\la a\ra $ to denote any number $b$ such that $|\frac ba -1|\ll 1 $ (in other words, $b$ equals $a$ up to lower order corrections). For example, $e^{g_0} = \la1\ra $.\\
We shall also use the notation $\ul n = \{1,\ldots,n\} $ if $n\in\mb N $.\\[5pt]
\textbf{Fonts.} We warn the reader to carefully take note of different fonts used in equations, such as serif ($x,A$) versus sans serif ($\sfx,\sfA$), or greek ($\Lambda,\phi$) versus upgreek ($\Uplambda ,\upphi$).

\erem
\end{rem}

\noindent
\textbf{Acknowledgements.} I want to thank Gordon Slade and David Brydges for helpful discussions and suggestions, and the Isaac Newton Institute for Mathematical Sciences for support and hospitality during the programme ``Scaling limits, rough paths, and quantum field theory'' (EPSRC Grant Number ER/R014604/1), when work on this paper was undertaken. This work was supported in part by NSERC of Canada.

\section{The results of \gk}\label{sgk}
The main technical result of \gk concerns the \textit{effective action}, defined by
\begin{align}\label{eqeffa}
Z(\phi) &=  \int e^{- \sum_{\xi\in \Lambda_0}\frac 12 \psi_\xi(-\Delta\psi)_\xi + \frac{\nu_c} 2\psi_\xi^2 + \frac{g_0}4\psi_\xi^4   } \prod_{x\in\Lambda}\delta\big(\phi_x - (C\psi)_x\big)\prod_{\xi\in\Lambda_0}\ud\psi_\xi
\end{align}
Here, $\phi:\mb Z^4 / L^{N-n}\mb Z^4=:\Lambda\to \mb R $  is called the block spin field, and
\begin{align}\label{eqbso}
(C\psi)_x &= L^{-3n} \sum_{\xi\in\square_x} \psi_\xi 
\end{align}
is the block spin operator. In its definition, we have associated to any point in $\Lambda$ a cube $\square_x := L^n x + \square \subset \Lambda_0  $, where $L^nx $ is the equivalence class in $\Lambda_0 $ of $L^n $ times any of the representatives of $x$ in $\mb Z^4 $, and $\square$ are the equivalence classes of $[-\tfrac 12(L^n-1),\tfrac12 (L^{n}-1)]^4\cap \mb Z^4 $. Note that the cubes $\square_x,x\in\Lambda $ are a partition of $\Lambda_0 $. We shall from now on view $\Lambda $ as embedded in $\Lambda_0 $ in this way, and write $\square_X  = \cup_{x\in X}\square_x$ for $X\subset \Lambda $.

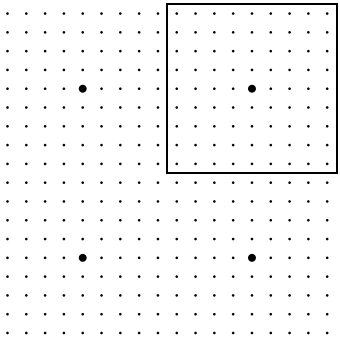
\begin{figure}
	\centering\vspace{15pt}
	\begin{tikzpicture}[scale=0.25]
	\foreach \i in {0,...,17}
	\foreach \j in {0,...,17}{
		\fill[black] (\i,\j) circle(2pt);
		\ifnum \i=4
		\ifnum \j=4
		\fill[black] (\i,\j) circle(6pt);
		\fi
		\fi
		\ifnum \i=4
		\ifnum \j=13
		\fill[black] (\i,\j) circle(6pt);
		\fi
		\fi
		\ifnum \i=13
		\ifnum \j=4
		\fill[black] (\i,\j) circle(6pt);
		\fi
		\fi
		\ifnum \i=13
		\ifnum \j=13
		\fill[black] (\i,\j) circle(6pt);
		\fi
		\fi
		\draw[draw=black] (8.5,8.5) rectangle ++(9,9);
	};
	\end{tikzpicture}\vspace{5pt}\caption{Fine dots: The elements $\xi\in\Lambda_0 $. Large dots: The elements $x\in\Lambda $. Square: The block $\square_x $ centered at a large dot $x$. ($L^n=9 $).}
\end{figure}

\noindent
The integer $n=0,\ldots,N $ is called the \textit{scale} of the effective action. We shall need the effective action only at a specific scale that is chosen carefully, depending on the magnetic field $h$. Indeed, calling the expression in the square brackets of (\ref{eqmagn}) $M(h) $, it follows from (\ref{eqeffa}) and (\ref{eqbso}) that
\begin{align}\label{eqmagn2}
M(h) &= \tfrac1{L^{n}|\Lambda|} \partial_\eps \log \int e^{ \sum_{x\in\Lambda}(L^{3n}h + \eps)\phi_x} Z(\phi)\prod_{x\in\Lambda} \ud\phi_x\bigg|_{\eps=0}
\end{align}
for any $n$. We shall choose $n$ such that the effective magnetic field
\begin{align}
\hbar := L^{3n} h
\end{align}
is not too small, see later for precise conditions.\\
We now explain some features of the representation of $Z(\phi) $ that is proven in \gk. Their argument implies bounds on the objects that are used in this representation, and these are stated in Theorem \ref{thmgk} below. Only the results summarized there will be used in the remaining sections. For more details, the reader is referred to the original paper.\\
The measure $\prod_{x\in\Lambda}\delta\big(\phi_x - (C\psi)_x\big)\prod_{\xi\in\Lambda_0}\ud\psi_\xi  $ appearing in (\ref{eqeffa}) is over a hyperplane of dimension $|\Lambda_0| - |\Lambda| $, and in \gk, this plane is parametrized as $\psi = \mc A\phi + \psi' $, where $\mc A\phi $ is a ``background field'' generated by the block spin field $\phi$, and $\psi' $ are fluctuating degrees of freedom that live in a linear space of dimension $|\Lambda_0| - |\Lambda| $.\\ 
The background field can be chosen completely explicitly in \gk, and is simply the minimizer of $\sum_{\xi\in \Lambda_0}  \psi_\xi(-\Delta\psi)_\xi $ with prescribed block averages $\phi $. Thus, the background field is linear in $\phi$, and it can be shown that the kernel $\mc A_{\xi,x} $ is translation invariant (more precisely, $\mc A_{\xi + L^ny,x+y} = \mc A_{\xi,x} $ for all $y$), and satisfies the bound
\begin{align}\label{boundmca}
|\mc A_{\xi,x} |&\leq c_{\mc A}^{-1}e^{-c_{\mc A}\, L^{-n}\,| L^nx - \xi| }
\end{align}
and the important properties
\begin{align}
\sum_{x\in\Lambda}\mc A_{\xi,x} &= 1 \label{mcai1}\\
\int_{\xi\in\square_y}\mc A_{\xi,x} &= \delta_{x,y}. \label{mcai2}
\end{align}
The constant $c_{\mc A} $ depends only on $L$.\\
The fluctuation field is treated by integrating out fluctuations on increasing length scales successively. The result may be written as a product of a Gaussian part, which has a natural expression directly in terms of the block spin field, and a non Gaussian part that is more naturally expressed in terms of the background field:
\begin{align}
Z(\phi) = \const e^{-\frac12 (\phi ,\bar G^{-1}\phi)} Z_{\geq 4}(\mc A \phi),\label{gkpsrep}
\end{align}
The quadratic form $\bar G^{-1} $ is a sum of three types of terms: First, a kernel obtained by explicit block spin transformations of the lattice Laplacian, times a ``wave function renormalization constant'' $z$ (which is close to 1); second, a mass renormalization counterterm, parametrized by a constant $\nu $, called the ``chemical potential'' (which is small and negative); and third, by contributions labeled ``quadratic irrelevant'' in the RG terminology. We will only need few details about these terms, which can be phrased in terms of the Fourier representation of $\bar G^{-1}$. Indeed, $\bar G^{-1}$ is translation invariant, and we may write it using the standard discrete Fourier transform
\begin{align}\label{deffourier}
\bar G^{-1}_{x,y} &= L^{4(n-N)}\sum_{p\in \Lambda^\star} \bar\mu(p)e^{ip(x-y)}
\end{align}
with $\Lambda^\star = \frac{2\pi}{L^{4(N-n)}}\mb Z^4 / 2\pi\mb Z^4 $. We then have that $\bar\mu(p)\approx z\, p^2 + \nu  $ as $p\to 0 $, and that $\bar\mu(p) $ is positive and bounded away from zero for $p $ sufficiently far away from $0$. The constants $\nu $ and $z$, as well as the ``RG irrelevant quadratic'' contributions to $\bar G^{-1}$ are constructed via a series of cluster expansions. See Theorem \ref{thmgk} for bounds on these numbers and objects. \\
The non Gaussian part $ Z_{\geq 4}(\psi)$ is also constructed via a series of cluster expansions. A representation that naturally emerges from such expansions is that of $Z_{\geq 4}(\psi) $ as a ``polymer gas with large field sets''. For this, let $D\subset \Lambda$  be a ``paved set'', i.e. a union of cubes 
\begin{align}\nonumber
\Delta\in \mc C:=\big\{\Delta\subset \Lambda \text{ an equivalence class of } \ttl\cdot m +  [-\tfrac12(\ttl-1),\tfrac12(\ttl-1)]\cap \mb Z^4 \text{ for some }m\in\mb Z^4\big\}.
\end{align}
Here $\ttl$ is a power of $L$ (it can be chosen as large as we like if $g_0$ is small enough). $D$ will serve as the set where the fields $\psi $ are allowed to become large, see below for precise definitions. For each fixed $D$, $ Z_{\geq 4}(\psi) $ enjoys a ``polymer gas representation with large field set $D$''. To define this notion, write
\begin{align}
\mc P(X) &= \Big\{\{X_1,\ldots,X_n\} =:\{X_m\}_1^n ,\, X_m = \text{union of cubes }\Delta,\nnb &\qqquad X_l\cap X_m = \emptyset \text{ if }l\neq m,\, \cup_{m=1}^nX_m = X  \Big\}
\end{align}
for the set of (unordered) partitions of $X\subset \Lambda$, of any length, into paved sets as above. In the following, we shall only encounter paved subsets of $\Lambda $, but will usually not emphasize this fact (for example, we will sloppily write $\sum_{X\subset\Lambda} $ for the sum over paved subsets of $\Lambda$). For a fixed paved set $D\subset\Lambda $, we also write
\begin{align}
\mc P_D(X) &= \big\{\{X_m\}_1^n\in\mc P(X)\text{ s.t. }  X_m\supset  \overline{X_m\cap D} \big\}
\end{align} 
where $\overline X = \{\Delta\in\mc C, d(\Delta,X)\leq 1\} $ (by convention, $\overline\emptyset=\emptyset $). Then, a polymer gas with large field set $D$ takes the general form 
\begin{align}
\sum_{\{X_m\}_1^n\in\mc P_D(X)} \prod_{m=1}^n g^D(X_m;\psi|_{\square_{X_m}})
\end{align}
The functions $g^D(X;\psi|_{\square_X}) $ are referred to as ``activities''. Expressions of this type are amenable to cluster and Mayer expansion arguments. \\
The effective action $Z_{\geq 4}(\psi) $ constructed in \gk satisfies, for any fixed $D$, the following polymer gas representation with large field set $D$:
\begin{align}\label{gkpsrep2}
Z_{\geq 4}(\psi)&= e^{-\frac g4 \int_{\xi \in \square_{D^c}} \psi_\xi ^4}\sum_{\{X_m\}_1^n\in\mc P_D(\Lambda)} \prod_{m=1}^n g^D(X_m;\psi|_{\square_{X_m}}) \qquad \forall\, \psi\in\mc D(D,\Lambda).
\end{align}
Here $g$ is an ``effective coupling constant'' (which is small and positive), and the activities $g^D(X;\psi|_{\square_X}) $ satisfy certain bounds, see Theorem \ref{thmgk} for details. We also used the notation $\int_{\xi\in Z} := L^{-4n}\sum_{\xi\in Z} $ for $Z\subset \Lambda_0 $.\\ 
The above equality holds if $\psi\in \mc D(D,\Lambda) $ belongs to the set of fields small outside $D$. This set is defined as follows: 
\begin{align}\label{deflfs}
\mc D(D,X) &= \Big\{\psi = \mc A\phi + \psi'\big|_{\square_X},\, \psi'\in\mc K(X), \phi\in\mb R^{\Lambda} \text{ s.t. } D(\phi) \subset D   \Big\},
\end{align}
where the set of small fields is (for $X\subset\Lambda$)
\begin{align}\label{defsfs}
\mc K(X) &= \Big\{\psi = \mc A (\phi+\phi')\big|_{\square_X} ,\, \phi\in\mb R^{\Lambda},\phi'\in \mb C^{\Lambda},\\\nonumber&\qqquad D(2\phi)\cap X=\emptyset,\,|\phi'_x|\leq c_{\mc K}\,c_2\, g^{-\frac14}\, e^{\half c_{\mc A} d(x,X) }\Big\}
\end{align}
and $D(\phi)$ is the large field set defined by $\phi\in\mb R^\Lambda$, namely the smallest paved $D\subset\Lambda $ such that
\begin{align}\label{defdphi}
|(\mc A\phi)_\xi|&\leq  c_2\ g^{-\frac14} e^{\frac1{10}c_g\, L^{-n} d(\xi,\square_{D^c}) }
\end{align}
The above definitions depend on constants $c_{\mc K} = c_{\mc K}(L) $ and $c_2 $ which are discussed in Theorem \ref{thmgk} below. \\
Note that, by (\ref{defdphi}), the field $\psi_\xi = (\mc A\phi)_\xi $ is not allowed to become bigger than $c_2\, g^{-\frac14} $ in the complement of the large field set $D$. In particular, the first factor in (\ref{gkpsrep2}) is always bounded below by $e^{-\frac{c_2^4}{4}|D^c|} $. \\
The decay of $Z_{\geq 4}(\psi) $ for larger fields is not quartic, and is encoded in bounds on $g^D(X;\psi|_{\square_X}) $. Roughly speaking, $g^D(X;\psi|_{\square_X}) $ decays as $e^{c_1|D\cap X|_\ell -\int_{\xi\in D}g^\half \psi_\xi^2} $ at large fields, where $c_1 $ is big (depending on $\ell$), but small compared to $c_2$, and we wrote
\begin{align}
|X|_\ell = \text{number of blocks }\Delta\subset X.
\end{align}
In addition to this decay, $g^D(X;\psi|_{\square_X}) $ is small whenever $X$ contains many blocks $\Delta $, or if these blocks are far apart. Both can be expressed by saying that $g^D(X;\psi|_{\square_X}) $ decays as $e^{-c_g \mc L(X)} $, where $c_g$ depends only on $L$, and 
\begin{align}
\mc L(X) &=  \text{length of shortest tree on the centers of }\Delta\subset X.
\end{align}
If $X=\Delta $ consists of a single block, and $ D\cap \Delta = \emptyset$, then $g^{D}(\Delta;\psi_{\square_\Delta}) $ is close to $1$.\\
There is some redundancy in the family of representations (\ref{gkpsrep2}) since, if $\psi\in \mc D(D,\Lambda) $, then also $\psi\in \mc D(D',\Lambda) $ for any $D'\supset D $ (in words, if $\psi $ is small outside $D$, then it is also small outside $D'$ if $D'\supset D $). By analyticity, this implies the following relation between activities $g^D $ and $g^{D'} $: 
\begin{align}\label{ancont}
\text{If } D\cap X = D'\cap X  \text{ with } D\subset D' \text{, then } g^D(X;\psi) = g^{D'}(X;\psi)  \text{ on } \mc D(D,X) .
\end{align}
We can now state the result of \gk:

\begin{thm}[\gk]\label{thmgk}
For any choice of $\ttl\geq \ttl(L) $, and any choice of $c_1\geq c_1(\ttl),c_2\geq c_2(c_1) $, if $g_0$ is small enough depending on these choices, the following holds: There exists a $\nu_c = \nu_c(g_0) $ such that, for all $n=0,\ldots,N $, the effective action $Z(\phi) $ of (\ref{eqeffa}) can be decomposed as in (\ref{gkpsrep}). Furthermore,
\begin{enumerate}
	\item The quadratic form $\bar G^{-1} $ in that decomposition is defined via its Fourier transform $\bar\mu(p) $ as in (\ref{deffourier}), and $\bar\mu(p) $ is the restriction to $\Lambda^\star $ of an entire periodic function on $\mb R^4 / 2\pi\mb Z^4  $, and there exist constants $c_\mu,c_\mu'>0$ and $ C_\mu<\frac1{2c_\mu} $, depending only on $L$, such that $\bar\mu(p)\geq c_\mu' $ if $|p|\geq c_\mu $, and 
	\begin{align}\label{smallpasympt}
	|\bar\mu(p) - z\, p^2 - \nu|\leq C_\mu \, |p|^3\qquad \text{if}\qquad |p|<c_\mu 
	\end{align}
	Here, $z$ and $\nu$ are the effective wave function renormalization and chemical potential.
	\item The function $Z_{\geq 4}(\psi) $ in the decomposition enjoys the representations (\ref{gkpsrep2}), with $g $ the effective coupling constant, and with activities satisfying the bounds
	\begin{align}\label{lfbound}
	\sup_{\psi\in   {\mc D}(D,X)}|g^D(X;\psi) | &\leq e^{c_1 |D\cap X|_\ttl - c_g\mc L(X) - \int_{\xi\in \square_{D\cap X}} g^{\frac 12 }|\psi_\xi|^2 + 20 \, g\, (\Im \psi_\xi)^4  }
	\end{align}
	Further, if $X=\Delta\in\mc C $ is a cube, and $D\cap X=\emptyset $, we have
	\begin{align}\label{blocksfbound}
	\sup_{\psi\in\mc D(D,\Delta)}|  g^D (\Delta;\psi) -1 |&\leq  g^{\frac13}
	\end{align}
	\item The effective coupling constant, the effective chemical potential, and the wave function renormalization satisfy
	\begin{align}
	g &= \big(g_0^{-1} + \tfrac{9\log L}{16\pi^2}\,  n + O(\log  n)\big)^{-1} \nnb
	|\nu| &\leq c_\nu \, g\label{eqpert}\\
	|z-1|&\leq c_z\, g.\nonumber
	\end{align}
	The constants depend only on $L$.
\end{enumerate}

\end{thm}

\begin{rem}
The results of \gk are actually more detailed than what is stated in Theorem \ref{thmgk}, since an inductive proof of Theorem \ref{thmgk} would probably not be possible without such details. We have given a condensed account and streamlined the notation to what we need for our current purposes. In particular:
\begin{enumerate}
	\item Detailed momentum space formulas and bounds are available for the relevant, marginal, and for the irrelevant contributions to the quadratic part of the effective action. While they have to be treated separately in the inductive proof of \gk, we have here collected them into the single quadratic form $(\phi,\bar G^{-1}\phi) $. It is straightforward to deduce from \gk the properties of $\bar\mu(p) $ stated in Theorem \ref{thmgk}, i.e. positivity for large enough momenta, and that the low momentum behavior is described by the relevant and marginal terms, up to irrelevant errors.
	\item The set $\mc P_D(X) $ of partitions over which we sum in the representation (\ref{gkpsrep2}) is slightly different in \gk than here. Note that in our definition, a polymer $X_m $ has to contain all neighboring cubes of any cube $\Delta\subset X_m\cap D $; equivalently, $X_m\cap D $ has to be a union of connected components of $D$, and $X_m$ must also contain all neighboring cubes of these connected components. By way of contrast, \gk only require that $X_m\cap D $ is a union of connected components of $D$. Every partition in the \gk sense defines a partition in our sense (simply merge neighboring polymers if they touch the same connected component of $D$), and so their representation of $Z_{\geq 4} $ may be resummed to give our representation. It is easy to see that the resummed activities satisfy the same bounds as in \gk, modulo a change in the constants.
	\item In \gk, a distinction is made for $g^D(X;\psi) $ (which is called $g^{nD}_X(\psi)$ there) between the cases $D\cap X=\emptyset $ (the small field activities) and $D\cap X\neq\emptyset $ (the large field activities). More details are needed for the small field case than what we discuss here, in particular, it is necessary for the terms contributing to $g^D(X;\psi) $ to be irrelevant in the RG sense. We will only need the bounds (\ref{lfbound}), (\ref{blocksfbound}), which are easy consequences of the detailed properties stated in \gk.
	\item The domain $\mc K(X) $ of complex valued small fields, and therefore also $\mc D(D,X) $, can actually be chosen larger than is done in (\ref{defsfs}). In fact, in \gk, $\mc K(X) $ is defined in terms of bounds on the field $\psi $ that are satisfied by any field of the form $\mc A(\phi+\phi')|_{\square_X} $, with $\phi,\phi'$ as in (\ref{defsfs}), by way of the bound (\ref{boundmca}) on the kernel $\mc A$ and similar bounds on its discrete derivatives. In our argument, we will only ever encounter fields of the form $\mc A(\phi+\phi')|_{\square_X} $, and such details are not important for us.
\end{enumerate}
\erem
\end{rem}

\begin{rem}\label{remivl}
It is also proven in \gk that the kernels $\bar G^{-1} $ and $ \mc A$ are the periodizations (with period $L^N $) of infinite volume kernels ($N=\infty$) satisfying the same bounds, and that, upon identifying the lattices $\Lambda$ for different values of $N$ in the natural way, the activities $\bar g(X) $ (for finite $X$) converge to limiting quantities, as $N\to\infty$. The limiting quantities satisfy the same bounds, with large and small field domains defined in the same way as above, except that $\phi $ in (\ref{deflfs}) and (\ref{defsfs}) is required to have finite support. We will need this fact to prove convergence of the thermodynamic limit in (\ref{eqmagn}).
\erem\end{rem}

\section{Heuristic derivation of the logarithmic correction}\label{remheur}

We now give a heuristic derivation of (\ref{logcorr2}). First, note that Theorem \ref{thmgk} suggests that
\begin{align}\label{zapprox}
Z(\phi) &\approx \const e^{-\half \sum_{x\in\Lambda} \phi_x(\bar G^{-1}\phi)_x - \frac g4 \int_{\xi\in\Lambda_0}(\mc A\phi)_\xi^4}
\end{align}
Indeed, the factor $e^{-c_g \mc L(X)} $ in the bound (\ref{lfbound}) is small (exponentially in $\ell$) whenever $X$ consists of more than one block $ \Delta$, so that $g^D(X;\psi) \approx \delta_{|X|_\ell,1} $ should hold by (\ref{blocksfbound}), at least in the small field region where
\begin{align}\label{remsfr}
|\phi_x|&\leq c_2 \, g^{-\frac14}\qquad\qquad\forall x\in\Lambda.
\end{align}
Since
\begin{align}
\sum_{\{X_m\}_1^n\in\mc P(\Lambda)}\prod_{m=1}^n \delta_{|X_m|_\ell,1} = 1,
\end{align}
it follows from these considerations that (\ref{zapprox}) should be accurate for all fields in the small field region (\ref{remsfr}).\\
For the magnetization (\ref{eqmagn2}), using that $\sum_{x\in\Lambda}\phi_x =\int_{\xi\in\Lambda_0}(\mc A\phi)_\xi  $ by (\ref{mcai2}), we conclude,
\begin{align}\label{approxmagnh}
M(h) &\approx \tfrac1{L^n|\Lambda|} \partial_\eps \log \int e^{-\half\sum_{x\in\Lambda}\phi_x(\bar G^{-1}\phi)_x - \int_{\xi\in\Lambda_0} -(\hbar + \eps)(\mc A\phi)_\xi  + \frac g4 (\mc A\phi)_\xi^4}\prod_{x\in\Lambda}\ud \phi_x\Big|_{\eps=0}.
\end{align}
\para{Shifting the field.}{fsheur} We would like to view the above integral as approximately Gaussian, and for that reason need to center it around the maximum of the argument of the exponential, at $\eps=0$. Even though the Hessian of that argument is not negative definite (recall (\ref{smallpasympt}) and $\nu<0 $), it turns out to be sufficient to look for the maximizer among constant fields $\phi_x\equiv \bar\phi $. Since $\mc A\bar\phi = \bar\phi $ by (\ref{mcai1}), $\bar\phi $ has to minimize
\begin{align}\label{defbarphi}
\bar\phi = \text{argmin }\{\tfrac \nu2 \bar\phi^2 + \tfrac g4 \bar\phi^4 - \hbar\bar\phi\}
\end{align}
With this choice of $\bar\phi $, upon translating fields in (\ref{approxmagnh}), we have
\begin{align}\label{approxmagnh2}
M(h)&\approx L^{-n}\bar\phi +  \tfrac1{L^n|\Lambda|} \partial_\eps \log \int e^{-\half\sum_{x\in\Lambda}\phi_x( G^{-1}\phi)_x - \int_{\xi\in\Lambda_0} V_\eps((\mc A\phi)_\xi) }\prod_{x\in\Lambda}\ud \phi_x\Big|_{\eps=0}
\end{align}
where
\begin{align}
G^{-1} &= \bar G^{-1} + 3g\bar\phi^2\, \mc A^\star\mc A \\\label{defveps}
V_\eps(x) &= -\eps \, x + g\bar\phi\, x^3 + \tfrac g4 x^4   
\end{align}
Note that (\ref{approxmagnh2}) is only accurate\footnote{Outside the small field region, the coupling constant $g$, and certainly (\ref{eqpert}), on which our heuristics are based, ultimately loose their meaning.} if the translation approximately preserves the small field region, that is, if
\begin{align}\blabel{sfcondshift}
|\bar\phi|\ll  g^{-\frac14}.
\end{align}
On the other hand, since $\bar G^{-1}\geq \nu = - \const g  $, $G^{-1} $ will be positive definite (see the Appendix for a proof), and could thus serve as the covariance of a Gaussian measure, as soon as 
\begin{align}\label{eqbarphibig}
|\bar\phi|\gg 1,
\end{align}
since this implies $3g\bar\phi^2\gg|\nu| $. As can be seen from (\ref{defbarphi}), the above conditions on $\bar\phi$ can be achieved by choosing an appropriately large $n$, depending on $h$. Indeed if $\frac{\hbar}{g} $ is large, we conclude from (\ref{defbarphi}) that
\begin{align}\label{approxphi}
\bar\phi\approx \Big(\tfrac \hbar g\Big)^{\frac 13} = L^n \Big(\tfrac h {g}\Big)^{\frac 13} 
\end{align}
and so, using (\ref{eqpert}) and $n\gg g_0^{-1} $ if $h$ is small, $1\ll|\bar\phi|\ll   g^{-\frac14} $ is essentially equivalent to
\begin{align}\label{stabregn}
n = \tfrac 13\log_L h^{-1} - x& \qquad\text{with}\\\nonumber (\tfrac13-\tfrac14) \log_L\log h^{-1}   + x'\leq x&\leq  \tfrac13 \log_L\log h^{-1}-x''  ,
\end{align}
where $x',x''\gg 1  $ can be chosen independently of $h$. Note that we then have
\begin{align}\label{hovergheur}
\tfrac h {g} &\approx \tfrac{h\cdot 9\,n\log L } {16\pi^2}\approx \tfrac{3h\log h^{-1}} {16\pi^2}  ,
\end{align}
and so the first term in (\ref{approxmagnh2}) looks exactly like what we want, i.e. like (\ref{logcorr2}).\footnote{In fact, the first term in (\ref{approxmagnh2}) looks like (\ref{logcorr2}) whenever $n = \tfrac13\log_L h^{-1}+x $ with $|x|\leq \const\delta\log_L h^{-1}$.}\\

\para{Perturbation theory for the remainder.}{ptheur} To complete the heuristics, we need to argue that the second term in (\ref{approxmagnh2}) is actually smaller than the first, for some choice of $n$ in the region (\ref{stabregn}). We shall attempt to do so by regarding it as a perturbation of a Gaussian measure. After diagonalizing the quadratic form $(\phi,G^{-1}\phi) $, we get that $M(h) = L^{-n}\bar\phi + M'(h) $, where 
\begin{align}\label{mpapprox}
M'(h)&\approx  \tfrac1{L^n|\Lambda|} \partial_\eps \log \int e^{- \int_{\xi\in\Lambda_0} V_\eps((\mc A\Gamma\upphi)_\xi) }\ud\mu(\upphi)\Big|_{\eps=0},
\end{align} 
with $\mu $ the standard Gaussian measure on $\mb R^\Lambda $ and
\begin{align}
\Gamma = G^\half.
\end{align}
A rather rigorous way to get an understanding of (\ref{mpapprox}), which contains in particular second order perturbation theory and gives the opportunity to discuss some essential stability aspects of that integral, is to study the following quantities, and their derivatives at $\eps=0$:
\begin{align}
\sfB(\square_\sfx) &= \int e^{-\int_{\xi\in\boxdot_\sfx}V_\eps((\mc A\Gamma\upphi)_\xi) } \ud\mu(\upphi|_{\square_\sfx})\\
 \sfB(\square_\sfx,\square_\sfy) &=  \int_0^1\ud s\int e^{-\int_{\xi\in\boxdot_\sfx}V_\eps((\mc A\Gamma(\upphi|_{\square_\sfx}+s\upphi|_{\square_\sfy}))_\xi) }\label{ce2ndterm}\\ \nonumber&\qquad\qquad\times \Bigg[\sums{x\in\square_\sfx\\x'\in\Lambda} \Gamma_{x',x}\upphi_x\partial_{\phi_{x'}}\Bigg]e^{-\int_{\xi\in\boxdot_\sfy}V_\eps((\mc A\phi)_\xi) }\bigg|_{\phi = \Gamma(\upphi|_{\square_\sfy}+s\upphi|_{\square_\sfx})} \ud\mu(\upphi|_{\square_\sfx\cup\square_\sfy})
\end{align}
Here, for some large, fixed $\mc L\in\ell\cdot \mb N $, $\Uplambda = \mc L\mb Z^4/L^{N-n}\mb Z^4 $, which we regard as a subset of $\Lambda $, and $\square_\sfx\subset\Lambda $ (resp. $\boxdot_\sfx:= \square_{\square_\sfx}\subset \Lambda_0 $) are the boxes of size $\mc L^4 $ (resp. $\mc L^4 L^{4n} $) centered at $\sfx\in\Uplambda $. Note that $\sfB(\square_\sfx),\sfB(\square_\sfx,\square_\sfy) $ are integrals of a fixed dimension, as opposed to the integral in $M'(h)$, whose dimension grows in the thermodynamic limit. Note also that $\sfB(\square_\sfx,\square_\sfy)  $ is one of two identical terms that appear when expressing 
\begin{align}
\int &e^{-\int_{\xi\in\boxdot_\sfx\cup\boxdot_\sfy}V_\eps((\mc A\Gamma(\upphi|_{\square_\sfx\cup\square_\sfy} )_\xi)} \ud\mu(\upphi|_{\square_\sfx\cup\square_\sfy}) - \sfB(\square_\sfx)\sfB(\square_\sfy)\nnb
&=\int  e^{-\int_{\xi\in\boxdot_\sfx}V_\eps((\mc A\Gamma(\upphi|_{\square_\sfx}+ s\upphi|_{\square_\sfy}))_\xi)}e^{-\int_{\xi\in\boxdot_\sfy}V_\eps((\mc A\Gamma(\upphi|_{\square_\sfy}+ s\upphi|_{\square_\sfx}))_\xi)}\ud\mu(\upphi|_{\square_\sfx\cup\square_\sfy})\bigg|_{s=0}^1\nonumber
\end{align}
using the Taylor formula. Quantities like $\sfB(\square_\sfx) $ and $\sfB(\square_\sfx,\square_\sfy)  $ appear in versions (see, e.g., \cite{doi:10.1063/1.4922014}) of a cluster expansion for $M'(h) $, using blocks of size $\mc L^4 $, and one expects
\begin{align}\label{mpapprox2}
M'(h) &\approx L^{-n}\mc L^{-4}\partial_\eps\Big[\log \sfB(\square_\sfx) + 2\sum_{\sfy\in\Uplambda}\frac{\sfB(\square_\sfx,\square_\sfy)}{\sfB(\square_\sfx)\sfB(\square_\sfy)}  \Big]_{\eps=0}
\end{align}
(for any $\sfx $, by translation invariance).\\

\para{Block spin variables.}{bsv} To analyze $\sfB(\square_\sfx) $ and $\sfB(\square_\sfx,\square_\sfy) $, we will need to use a version of block spin coordinates local to each of the blocks $\square_\sfx,\square_\sfy $, and we describe these variables first. We enumerate the $\mc L^4 $ points in each block $\square_\sfx $ arbitrarily, so that $\square_\sfx =\{x_{\sfx,1},\ldots,x_{\sfx,\mc L^4}\} $. Block spin coordinates are then defined by
\begin{align}
\dot\upphi_{\sfx,k}&= \left\{\begin{array}{ll}  \mc L^{-4}\sum_{x\in\square_\sfx}\upphi_x & k=1 \\ \upphi_{x_{\sfx,k}}- \mc L^{-4}\sum_{x\in\square_\sfx}\upphi_x & k=2,\ldots,\mc L^4  \end{array}\right.\\
\upphi_{x_{\sfx,k}} &= \left\{\begin{array}{ll}  \dot\upphi_{\sfx,1} - \sum_{k>1}\dot\upphi_{\sfx,k}    & k=1  \\ \dot\upphi_{\sfx,k} + \dot\upphi_{\sfx,1}  & k>1 \end{array}\right.
\end{align}
We shall always think of $\upphi$ and $\dot\upphi$ as functions of each other, even if this is not obvious from the notation. We make the following observations:
\begin{enumerate}
	\item For any $\sfX\subset\Uplambda $, 
	\begin{align}\label{l2dots}
	\sum_{x\in\square_\sfX}\upphi_x^2 &= \sum_{\sfx\in\sfX} \Big[\mc L^4 \dot\upphi_{\sfx,1}^2 + \sum_{k>1}\dot\upphi_{\sfx,k}^2 + \Big(\sum_{k>1}\dot\upphi_{\sfx,k}\Big)^2\Big]
	\end{align}
	In particular, if we define 
	\begin{align}\label{defsfXnorm}
	\Vert\upphi\Vert_{\sfX}^2  &=  \sum_{\sfx\in\sfX} \Big[\mc L^4\dot\upphi_{\sfx,1}^2 + \mc L^{-4}\sum_{k>1}\dot\upphi_{\sfx,k}^2 + \mc L^{-4}\Big(\sum_{k>1}\dot\upphi_{\sfx,k}\Big)^2 \Big],
	\end{align}
	we have that 
	\begin{align}\label{defncal}
	\int e^{\alpha\Vert\upphi\Vert_{\square_\sfx}^2} \ud\mu(\upphi|_{\square_\sfx}) = (1-2\alpha)^{-\frac12}\cdot\big(1-\tfrac{2\alpha}{\mc L^4}\big)^{-\frac{\mc L^4-1}{2}} =:\mc N
	\end{align}
	is bounded uniformly in $\mc L^4$, as long as $\alpha<\frac12 $.
	\item We note that
	\begin{align}\label{cauchycv}
	|\dot\upphi_{\sfx,1}|\leq   r,\, |\dot\upphi_{\sfx,k}|\leq \mc L^{-4}   r \qquad\text{implies}\qquad |\upphi_x|\leq 2  r \quad\forall x\in\square_\sfx,\,   r\geq 0.
	\end{align}
	\item We have
	\begin{align}\label{cvds}
	\sum_{x,x'\in\Lambda}\Gamma_{x',x}\upphi_x\partial_{\phi_{x'}}  &= \sums{\sfx,\sfy\in\Uplambda\\j,k\in\ul{\mc L^4}}\ddot \Gamma_{\sfx,j;\sfy,k}\dot\upphi_{\sfy,k}\partial_{\dot\phi_{\sfx,j}}
	\end{align}
	with
	\begin{align}
	\ddot \Gamma_{\sfx,j;\sfy,k} &= \mc L^{-4}\sum_{x\in\square_\sfx}\left\{\begin{array}{ll}   \sum_{y\in\square_\sfy} \Gamma_{x,y}  &j=k=1  \\    \Gamma_{x,x_{\sfy,k}} - \Gamma_{x,x_{\sfy,1}}   & j=1,k>1  \\
	\sum_{y\in\square_\sfy } \Gamma_{x_{\sfx,j},y}- \Gamma_{x,y} & j>1,k=1 \\    \Gamma_{x_{\sfx,j},x_{\sfy,k}} - \Gamma_{x_{\sfx,j},x_{\sfy,1}}  - \Gamma_{x,x_{\sfy,k}} + \Gamma_{x,x_{\sfy,1}}   &  j,k>1 \end{array}\right. 
	\end{align}
	\item We have
	\begin{align}
	(\mc A\Gamma\upphi)_{\xi} &= \sum_{\sfy\in\Uplambda,k\in\ul{\mc L^4}}\dot {{\mc A}\Gamma}_{\xi;\sfy,k}\dot\upphi_{\sfy,k} 
	\end{align}
	with
	\begin{align}
	\dot {{\mc A}\Gamma}_{\xi;\sfy,k} &=    \left\{ \begin{array}{ll}   \sum_{y\in\square_\sfy}(\mc A\Gamma)_{\xi,y}   & k=1  \\   (\mc A\Gamma)_{\xi,x_{\sfy,k}} - (\mc A\Gamma)_{\xi,x_{\sfy,1}}   & k >1   \end{array}  \right. 
	\end{align}
\end{enumerate}

\para{Activity on a single block.}{sbheur} We now analyze $\sfB(\square_\sfx) $. It will turn out to be sufficient to use crude bounds on zeroth order perturbation theory, that is, on
\begin{align}\label{eq0optb}
\sfB(\square_\sfx) &= 1+ \int_0^1 \ud s\, \int \Big[ \int_{\xi\in\boxdot_\sfx}V_\eps((\mc A\Gamma\upphi)_\xi) \Big]  e^{-s\int_{\xi\in\boxdot_\sfx}V_\eps((\mc A\Gamma\upphi)_\xi) }\ud\mu(\upphi|_{\square_\sfx})
\end{align}
We would like to use a quadratic lower bound on $V_\eps(x) $ to estimate the remainder integral by a Gaussian-type one, which would be easy to evaluate. There is, of course, a two parameter family of quadratic lower bounds on $V_\eps(x) $, but here we will base our analysis on\footnote{In constructions exactly at the critical point, a constant lower bound is sufficient. The stability aspect of our problem differs from that, because $\inf_{|\eps|\leq \upepsilon  ,|r|\leq r_0,x\in\mb R } \Re V_\eps(x+r) \sim -g\bar\phi^4 $, which, while being small in the region (\ref{stabregn}), turns out to be not small enough, and would create nonperturbatively large factors in the cluster expansion.} 
\begin{align}\label{eqbasicstab}
\inf_{|\eps|\leq \upepsilon  ,|r|\leq 2\mf r } \Re V_\eps(x+r) \geq -\big\la g\bar\phi^2 \big\ra \, x^2 - 3\la \mf r(\upepsilon + 12g\bar\phi \, \mf r^2)\ra 
\end{align}
which holds whenever $x\in\mb R $ and 
\begin{align}\blabel{stabcond} 
\upepsilon \ll g\bar\phi^2\,\mf r\;,\qquad \mf r\ll \bar\phi
\end{align}
(We remind the reader of the notation described in Remark \ref{remsmallnot}). The proof of (\ref{eqbasicstab}) basic.\\
Next, we need to understand whether the bound
\begin{align}\label{gaussbound}
e^{-s\int_{\xi\in\boxdot_\sfx}V_\eps((\mc A\Gamma\upphi)_\xi) }\leq e^{\la g\bar\phi^2\ra\int_{\xi\in\boxdot_\sfx}(\mc A\Gamma\upphi)_\xi^2 +3\mc L^4\la \mf r(\upepsilon + 12g\bar\phi \, \mf r^2)\ra   }
\end{align}
is, at all, integrable with respect to the standard Gaussian. For this, note that, by the Schur bound on the spectral radius, for any two sets $\sfX,\sfX'\subset\Uplambda $ and any $\upphi $ with support in $\square_{\sfX'} $,
\begin{align}\label{schurbound}
\int_{\xi\in \boxdot_\sfX}(\mc A\Gamma\upphi)_\xi^2 &\leq \Vert \mc A\Gamma\Vert_{\text{HS}(\sfX,\sfX')}^2 \Vert \upphi\Vert_{\sfX'}^2
\end{align}
with 
\begin{align}
\Vert \mc A\Gamma\Vert_{\text{HS}(\sfX,\sfX')}^2 &= \sup_{\sfx\in\sfX',k}\sum_{\sfy\in\sfX',k'}\mc L^{-2\delta_{k=1}-2\delta_{k'=1} + 2\delta_{k>1}+2\delta_{k'>1}}\sum_{\sfz\in\sfX}\Bigg| \int_{\xi\in \boxdot_\sfz}\dot{{\mc A}\Gamma}_{\xi;\sfx,k}\dot{{\mc A}\Gamma}_{\xi;\sfy,k'}\Bigg|
\end{align}
In the Appendix, we will prove 
\begin{enumerate}
	\item The bound 
	\begin{align}\label{roughboundcov}
	\Vert \mc A\Gamma\Vert^2_{\text{HS}(\sfX,\sfX')}\leq\const \mc L^{12}\cdot |\sfX|\cdot |\sfX'| .
	\end{align}
	This is an immediate consequence of the pointwise boundedness of $\Gamma $ and the exponential decay of $\mc A $ (see (\ref{boundmca})), but it does not rely on the decay of $\Gamma$ and is therefore not uniform in $|\sfX|,|\sfX'| $.
	\item The somewhat subtle bound
	\begin{align}\label{sharpboundcov}
	\Vert \mc A\Gamma\Vert_{\text{HS}(\sfX,\sfX')}^2&\leq \la ( 3g\bar\phi^2)^{-1}\ra,
	\end{align}
	which holds under the condition 
	\begin{align}\blabel{condhsnorm}
	\mc L^{24}g\bar\phi^2\log^2(g\bar\phi^2) \ll 1.
	\end{align}
	This bound relies on the momentum space regularity of $\Gamma$ (see (\ref{smallpasympt})), and thus uses the decay of $\Gamma $. Indeed, $(3g\bar\phi^2)^{-1} $ is essentially the zero momentum value of $\mc A\Gamma(\mc A\Gamma)^\star  = \mc A(\bar G^{-1}+3g\bar\phi^2\mc A^\star\mc A)^{-1}\mc A^\star $ by (\ref{smallpasympt}), (\ref{eqpert}), (\ref{eqbarphibig}), and (\ref{mcai1}), (\ref{mcai2})
\end{enumerate}
By (\ref{roughboundcov}), under the condition
\begin{align}\blabel{1blockcond1}
\mc L^{12}g\bar\phi^2&\ll 1,
\end{align}
the bound (\ref{gaussbound}) is integrable with respect to the standard Gaussian, and it follows readily from (\ref{eq0optb}) that
\begin{align}\label{bound1pact}
\big|\sfB(\square_x) -1 \big|\leq \const e^{3\mc L^4\la \mf r(\upepsilon + 12g\bar\phi \, \mf r^2)\ra}\,  \big[\mc L^4V_{\upepsilon} ( \mc L^{6}) \big] 
\end{align}
for some numerical constant, and whenever $|\eps|\leq \upepsilon $. This bound says that $\sfB(\square_\sfx)\approx1  $ (i.e. zeroth order perturbation theory) is accurate for block sizes that are small enough so that, in addition to (\ref{1blockcond1}),
\begin{align}\blabel{1blockcond2}
\mc L^4  \upepsilon\,  \mf r \ll1\;,\qquad \mc L^4  g\bar\phi \, \mf r^3 \ll 1
\end{align}
and
\begin{align}
\mc L^{10}\upepsilon \ll 1\;,\qquad \mc L^4V_{0} ( \mc L^{6})  \ll 1\blabel{1blockcond3}
\end{align}

\para{Activity on two blocks.}{tbheur} We now proceed to bound $\sfB(\square_\sfx,\square_\sfy) $. Inserting (\ref{cvds}) into (\ref{ce2ndterm}), and performing the derivative using the Cauchy formula and (\ref{cauchycv}), we get
\begin{align}
|\sfB(\square_\sfx,\square_\sfy)|&\leq \frac{\Vert \mc A\Vert \, \Vert \Gamma\Vert_\sfy }{\mf r}\sup_{\substack{s\in[0,1] \\ \Vert r\Vert_\infty\leq 2\mf r}} \Bigg|\int (1+\Vert\dot\upphi\Vert_{\{\sfx\}}^2) \,e^{-\int_{\xi\in\boxdot_\sfx}V_\eps((\mc A\Gamma(\upphi|_{\square_\sfx}+s\upphi|_{\square_\sfy}))_\xi)} \nnb&\qqquad\times e^{-\int_{\xi\in\boxdot_\sfy}V_\eps((\mc A\Gamma(\upphi|_{\square_\sfy}+s\upphi|_{\square_\sfx}))_\xi + r_\xi) } \ud\mu(\upphi|_{\square_\sfx\cup\square_\sfy})\Bigg|\label{1stboundb2block}
\end{align}
where
\begin{align}
\Vert \Gamma\Vert_\sfy &=\sup_{\substack{\sfx\in\Uplambda\\ k\in\ul{\mc L^4} }} \sums{\sfx'\in\Uplambda\\ j\in\ul{\mc L^4} }\mc L^{4\delta_{j>1} + 2\delta_{k>1} - 2\delta_{k=1}  }e^{-\half c_{\mc A} d(\square_{\sfx'},\square_\sfy) }\, |\ddot\Gamma_{\sfx',j;\sfx,k}|\\
\Vert \mc A\Vert &= \sup_{\xi\in\Lambda_0} \sum_{x\in\Lambda} e^{\frac12 c_{\mc A}L^{-n}|L^nx-\xi|}\, |\mc A_{\xi,x}|.\label{normacal}
\end{align}
Clearly, $\Vert\mc A\Vert\leq  \;$const, with a constant depending only on $L$. We have
\begin{align}
\int_{\xi\in\boxdot_\sfx}V_\eps((\mc A\Gamma(\upphi|_{\square_\sfx}+s\upphi|_{\square_\sfy}))_\xi) +  \int_{\xi\in\boxdot_\sfy}V_\eps((&\mc A\Gamma(\upphi|_{\square_\sfy}+s\upphi|_{\square_\sfx}))_\xi + r_\xi)\nnb &= \int_{\xi\in\boxdot_\sfx\cup\boxdot_\sfy}V_\eps((\mc A\Gamma^s\upphi|_{\square_\sfx\cup\square_\sfy})_\xi + r_\xi)
\end{align}
where
\begin{align}
(\mc A\Gamma^s)_{\xi,x} &= \left\{\begin{array}{ll} s(\mc A\Gamma)_{\xi,x} & \xi\in\square_\sfx\text{ and }x\in\square_\sfy \text{ or } \xi\in\square_\sfy\text{ and }x\in\square_\sfx \\ (\mc A\Gamma)_{\xi,x} & \text{else}  \end{array}  \right.
\end{align}
and we set $r_\xi =0 $ if $\xi\not\in\boxdot_\sfy $. Using again the Gaussian bound (\ref{eqbasicstab}), the Schur bound (\ref{schurbound}), and $\Vert \mc A\Gamma^s\Vert_{\text{HS}(\sfX,\sfX')}\leq \Vert \mc A\Gamma\Vert_{\text{HS}(\sfX,\sfX')} $, we get 
\begin{align}
\Re\int_{\xi\in\boxdot_\sfx}V_\eps((\mc A\Gamma(\upphi|_{\square_\sfx}+&s\upphi|_{\square_\sfy}))_\xi) +\Re \int_{\xi\in\boxdot_\sfy}V_\eps((\mc A\Gamma(\upphi|_{\square_\sfy}+s\upphi|_{\square_\sfx}))_\xi + r_\xi)\nnb &\geq -\big\la g\bar\phi^2\big\ra \Vert \,  \mc A\Gamma\Vert^2_{\text{HS}(\{\sfx,\sfy\},\{\sfx,\sfy\})} \Vert \upphi\Vert_{\{\sfx,\sfy\}}^2 - 6\mc L^4 \la \mf r(\upepsilon + 12g\bar\phi \, \mf r^2)\ra 
\end{align}
for any $s\in[0,1] $ and $\Vert r\Vert_\infty\leq \mf r $. We could use (\ref{roughboundcov}) to conclude that this gives an integrable bound for the integrand in (\ref{1stboundb2block}), but this would require a condition on $\mc L $ that is $4$ times stronger than the condition (\ref{1blockcond1}) for the single block integral. Instead, we shall take the opportunity to use a bound that is uniform in the number of blocks, and thus will also work for any term in the cluster expansion. Indeed, by (\ref{sharpboundcov}), 
\begin{align}\label{stabboundfields}
-\big\la g\bar\phi^2\big\ra \Vert \,  \mc A\Gamma\Vert^2_{\text{HS}(\{\sfx,\sfy\},\{\sfx,\sfy\})} \Vert \upphi\Vert_{\{\sfx,\sfy\}}^2 \geq - \la\tfrac 13\ra\,\Vert \upphi\Vert_{\{\sfx,\sfy\}}^2.
\end{align}
We make the important observation
\begin{align}
\frac 13<\frac 12,
\end{align}
and conclude from it that 
\begin{align}
|\sfB(\square_\sfx,\square_\sfy)| &\leq \const \frac{ \Vert\Gamma\Vert_\sfy}{\mf r},
\end{align}
which holds under the conditions (\ref{stabcond}), (\ref{condhsnorm}) and (\ref{1blockcond2}).\\

\para{Heuristics for the remainder.}{remainheur} The bounds derived in the previous two paragraphs can now be used to discuss heuristics for the remainder term $M'(h) $ of the magnetization, see (\ref{mpapprox}). One expects $M'(h) $ to be approximately given by (\ref{mpapprox2}), and we now bound the two terms on the right hand side of this equation, and compare the bounds with the leading contribution to $M(h)$, namely $L^{-n}\bar\phi $. For the first, by the bound on the activity on a single block, and analyticity in $\eps$,
\begin{align}
\big|\partial_\eps\log \sfB(\square_\sfx)\big|_{\eps=0} & = \big\la \partial_\eps \big(\sfB(\square_\sfx) -1 \big) \big|_{\eps=0} \big\ra \ll   \upepsilon^{-1}
\end{align} 
For the second term, we first bound
\begin{align}
\bigg| \sum_\sfy\frac{\sfB(\square_\sfx,\square_\sfy)}{\sfB(\square_\sfx)\sfB(\square_\sfy)}  \bigg|  &\leq \const \frac{ \sum_\sfy\Vert\Gamma\Vert_\sfy}{\mf r } 
\end{align}
We have
\begin{align}\label{entrbound}
\sum_\sfy\Vert\Gamma\Vert_\sfy \leq \const \Vert \Gamma\Vert
\end{align}
with an $L$ dependent constant and (we give a symmetric definition for later use)
\begin{align}\label{normgamma}
\Vert \Gamma\Vert &= \max\Big\{\sup_{\sfx\in\Uplambda} \sum_{\sfx'\in\Uplambda} \Vert \Gamma\Vert_{\sfx,\sfx'}, \sup_{\sfx'\in\Uplambda} \sum_{\sfx\in\Uplambda} \Vert \Gamma\Vert_{\sfx,\sfx'} \Big\}\nnb
\Vert \Gamma\Vert_{\sfx,\sfx'} &= \sup_{ k\in\ul{\mc L^4} } \sums{ j\in\ul{\mc L^4} }\mc L^{4\delta_{j>1} + 2\delta_{k>1} - 2\delta_{k=1}  } \, |\ddot\Gamma_{\sfx',j;\sfx,k}|
\end{align}
We will prove in the Appendix that, under the condition (\ref{condhsnorm}),
\begin{align}\label{basiccovestimate}
\Vert \Gamma\Vert &\leq \const \mc L^{-2}(g\bar\phi^2)^{-\half}.
\end{align}
Therefore,
\begin{align}
\bigg| \sum_\sfy\frac{\sfB(\square_\sfx,\square_\sfy)}{\sfB(\square_\sfx)\sfB(\square_\sfy)}  \bigg|  &\leq \const (\mc L^2\, \mf r\,  \bar\phi\, g^\half)^{-1}
\end{align}
We shall now impose the condition
\begin{align}\blabel{condconvce}
\mc L^2\, \mf r\,  \bar\phi\, g^\half\gg 1.
\end{align}
We shall later see that this condition implies convergence of the cluster expansion. It is crucial to note that it works in the opposite way to all conditions previously imposed, most prominently (\ref{1blockcond2}), which required $\mc L $ and $\mf r $ to be not too large. We will shortly check that all conditions can actually be satisfied. Assuming that this is the case, by analyticity in $\eps$, we get the bound
\begin{align}
\bigg|\partial_\eps \frac{\sfB(\square_\sfx,\square_\sfy)}{\sfB(\square_\sfx)\sfB(\square_\sfy)}  \bigg|_{\eps=0} &\leq\const (\mc L^2\, \mf r\,  \bar\phi\, g^\half\,\upepsilon)^{-1} .
\end{align}
We conclude
\begin{align}
|M'(h)| &\ll   L^{-n} \,\upepsilon^{-1}\, \mc L^{-4} ,
\end{align}
and this will be smaller than the leading contribution $L^{-n}\bar\phi $ if
\begin{align}\blabel{condeps}
\mc L^4\upepsilon \gg  \bar\phi^{-1}  
\end{align}
Again, this condition works in the opposite direction as the ones of the previous paragraphs, most prominently (\ref{1blockcond2}).\\

\para{Choice of parameters.}{cparaheur} We now show that the various conditions we needed to derive the previous bounds are consistent. We start by discussing (\ref{condconvce}) and the second conditions in (\ref{stabcond}) and (\ref{1blockcond2}). All three of them can be satisfied by choosing
\begin{align}\label{regr}
\big(\mc L^4 g \bar\phi^2\big)^{-\half}\ll \mf r\ll \big(\mc L^4 g \bar\phi\big)^{-\frac13},
\end{align}
which is possible whenever
\begin{align}\label{lfcondshift}
\mc L\bar\phi\gg g^{-\frac14}.
\end{align}
This important conclusion has to be contrasted with (\ref{sfcondshift}): While the latter demands that the ``magnetization scale'' $n$ at which we stop the symmetric flow of \gk be small enough so that the minimizer $\bar\phi$ of the effective potential still belongs to the small field region, the former demands that the minimizer should be a large field (though not quite as large as $c_2\,g^{-\frac14} $) at the scale $n + \log_L \mc L $ of the cluster expansion.\\
The conditions (\ref{condhsnorm}), (\ref{1blockcond1}) and the second condition in (\ref{1blockcond3}) all restrict the size of $\mc L $ to be not too large. Indeed, the latter is the strongest and is equivalent to
\begin{align}\label{upperboundlcal}
\mc L^{22}  g\bar\phi  \ll 1. 
\end{align}
Together with (\ref{lfcondshift}), they therefore give a lower bound on allowed choices for $n$, which complements the upper bound that follows from (\ref{sfcondshift}). More precisely, (\ref{lfcondshift}) and (\ref{upperboundlcal}) can be satisfied by choosing 
\begin{align}\label{regl}
(g^{\frac14}\bar\phi)^{-1} \ll \mc L \ll (g\bar\phi )^{-\frac1{22}},
\end{align}
which is possible whenever
\begin{align}\label{lbshift}
g\bar\phi^{\frac {14}3}\gg 1.
\end{align}
Note $\frac{14}3 >4 $. As in the discussion leading up to (\ref{stabregn}), the conditions (\ref{sfcondshift}) and (\ref{lbshift}) translate into the conditions
\begin{align}\label{regn}
n = \tfrac 13\log_L h^{-1} - x& \qquad\text{with}\\\nonumber (\tfrac13-\tfrac14) \log_L\log h^{-1}  + x'\leq x&\leq  (\tfrac13 - \tfrac3{14})  \log_L\log h^{-1}-x''  ,
\end{align}
(with $x',x''\gg 1  $) on the magnetization scale, and any choice of $n$, $\mc L $ and $\mf r $ compatible with (\ref{regn}), (\ref{regl}) and (\ref{regr}) will allow for the heuristics to be made rigorous.\\
Finally, the bound on the remainder term $M'(h) $ for the magnetization depends on the choice of the analyticity radius $\upepsilon $. This choice is restricted by the first conditions of (\ref{stabcond}), (\ref{1blockcond2}) and (\ref{1blockcond3}), and by (\ref{condeps}). All these are satisfied if 
\begin{align}\label{regeps}
\mc L^{-4} \bar\phi^{-1}\ll\upepsilon\ll \min\{\mc L^{-10}, \mc L^{-4}\mf r^{-1}\}.
\end{align}
Note that $\mc L^{-4} \bar\phi^{-1}\ll \min\{\mc L^{-10}, \mc L^{-4}\mf r^{-1}\}$ holds for any choice of $n,\mc L,\mf r $ made as above. \\
Optimal choices of the parameters $\upepsilon,\mf r, \mc L $ and $n$ allowed by the above discussion (namely, $\upepsilon $ as big as possible, $\mf r$ as small as possible, $\mc L\sim (g\bar\phi^2)^{-\frac1{16}} $ and $n $ such that $\bar\phi \sim g^{-\frac14} $) yield for the magnetization
\begin{align}\label{heurboundm}
|M(h) -L^{-n}\bar\phi| & = |M'(h)| \leq \const  L^{-n} \bar\phi\cdot g^{\frac1{16}},
\end{align}
with a constant that only depends on the constants implied in the $\ll$ symbols of the bounds (\ref{regr}), (\ref{regl}), (\ref{regn}) and (\ref{regeps}). The bound (\ref{heurboundm}) isolates the leading contribution to the magnetization and, together with (\ref{approxphi}) and (\ref{hovergheur}), concludes our heuristics.\\
In the rigorous argument of sections \ref{sce} and \ref{bounds}, there will be an additional constraint on $\mc L $ that is much stronger than the upper bound of (\ref{regl}), and keeps us from optimizing parameters as just explained. Instead of the factor $\const g^{\frac1{16}} $, we will only be able to achieve a constant $\delta $ that can be chosen as small as we like, as long as $g$ and $h$ are small enough. This implies Theorem \ref{mainthm}. Note that any factor $\const g^{\delta} $ with $\delta>0 $ would imply stronger asymptotics of the kind discussed in Remark \ref{remstrongsympt}.

\section{The cluster expansion for the magnetization}\label{sce}

We now turn to a rigorous analysis of (\ref{eqmagn2}), using the representation of Theorem \ref{thmgk}. We notice that (\ref{eqmagn2}), (\ref{gkpsrep}), a translation by the constant field $\bar\phi$, and a scaling of the translated field, like in section \ref{remheur}, imply
\begin{align}
M(h) &= L^{-n}\bar\phi  + M'(h)
\end{align}
with
\begin{align}\label{defmp}
M'(h) &= \tfrac{1}{L^n|\Lambda|} \partial_\eps\log \int  \tilde Z_{\geq 4}(\mc A\Gamma\upphi) \ud\mu(\upphi) \bigg|_{\eps=0} =: \tfrac{1}{L^n|\Lambda|} \partial_\eps \log Z|_{\eps=0}.
\end{align}
Here
\begin{align}
\tilde Z_{\geq 4}(\psi) &= e^{\frac g4 \bar\phi^4|\Lambda| + \int_{\xi\in\Lambda_0} (\hbar + \eps-\bar G^{-1}\bar\phi)\psi_\xi +\frac32 g\bar\phi^2\psi_\xi^2 }Z_{\geq 4}(\psi+\bar\phi).
\end{align}

\vspace{5pt}

\para{A polymer system representation}{inpsrep}
Our first step towards an analysis of $M'(h) $ is to write $Z$ of (\ref{defmp}) as the integral of a polymer system, with polymers of scale $\ell$, using the representations for $Z_{\geq 4} $ of (\ref{gkpsrep2}). This is done in the Proposition below. In its statement, we use the following objects: We define 
\begin{align}
\Psi(X) &= \{\uppsi = (\psi,\psi'),\, \psi = \mc A\phi|_{\square_X}\text{ for some }\phi\in\mb R^\Lambda,\,\psi'\in\mc K_{\mf r}(X) \}\\
\mc K_{\mf r}(X) &= \Big\{\psi = \mc A \phi\big|_{\square_X} ,\, \phi\in \mb C^{\Lambda},\, |\phi_x|<  2\,\Vert \mc A\Vert^{-1}\, \mf r \, e^{\half c_{\mc A} d(x,X) }\Big\}.
\end{align}
We shall also sometimes write $\mc A(\phi,\phi') := (\mc A\phi,\mc A\phi') $. Whenever $\uppsi $ is specified, we will refer to its two components by $ \psi,\psi'$ (instead of something like $\uppsi_1,\uppsi_2 $), and write $\uppsi_\xi = (\psi_\xi,\psi'_\xi) $. For $\upphi\in\mb R^\Lambda $, we also define
\begin{align}
\Psi(X,\upphi) &= \{\uppsi\in\Psi(X)\text{ such that } \psi\in \mc D(\upphi,X)\}.
\end{align}
Here,
\begin{align}
\mc D(\upphi,X) &= \{\psi = \mc A\Gamma\upphi'|_{\square_X},\, \upphi'\in\mb R^\Lambda,\,  \upphi'_x=\upphi_x\,\forall x\in X,\, D(2\Gamma(\upphi'|_{X^c}))\cap X = \emptyset \} .
\end{align}
Note that $\Psi(X,\upphi) $ depends only on $\upphi_x $ for $x\in X $. It is clear that
\begin{align}\label{lfinclproperty}
\mc D(\upphi,X)+\mc K_{\mf r}(X)+\bar\phi\subset \mc D(R(\upphi|_X),X) ,
\end{align}
where
\begin{align}
R(\upphi) &=  D(2\Gamma\upphi+\bar\phi)
\end{align}

\begin{prop}\label{proppsz}
We have
\begin{align}\label{psz}
Z &= \int  \sum_{\{X_m\}_1^n\in\mc P(\Lambda)}  \prod_{m=1}^n g(X_m;(\mc A\Gamma\upphi|_{\square_{X_m}},0);\upphi|_{X_m})  \ud\mu(\upphi)
\end{align}
where
\begin{align}\label{defg}
g(X;\uppsi;\upphi) &= e^{\frac g4 \bar\phi^4|R(\upphi|_X)| + \int_{\xi\in\square_{R(\upphi|_X)}} (\hbar + \eps-\bar G^{-1}\bar\phi)(\psi_\xi+\psi'_\xi) +\frac32 g\bar\phi^2(\psi_\xi+\psi'_\xi)^2 - \int_{\xi\in \square_{X\setminus R(\upphi|_X)}} V_\eps(\psi_\xi+\psi'_\xi)} \nnb &\qqquad \times g^{R(\upphi|_X)}(X;(\psi + \psi')|_{\square_X} + \bar\phi)\times \delta\big(X\supset \overline{ R(\upphi|_X)}\big)\\\nonumber&\qqquad\times \delta\big((\psi+\psi')|_{\square_X}+\bar\phi\in \mc D(R(\upphi|_X),X)\big).
\end{align}
depends only on $\uppsi_\xi,\upphi_x $ for $\xi\in\square_X,x\in X $. $g(X,\uppsi;\upphi)$ is smooth in $\psi $ and analytic in $\psi' $ if $\uppsi\in\Psi(X,\upphi) $, and in this domain we have $\partial_{\psi_\xi}g(X,\uppsi;\upphi) = \partial_{\psi'_\xi}g(X,\uppsi;\upphi) $ and similarly for higher derivatives.

\end{prop}

\begin{proof}
We need to show that, with $g(X;\uppsi;\upphi) $ defined as in (\ref{defg}), we have
\begin{align}
\sum_{\{X_m\}_1^n\in\mc P(\Lambda)}  \prod_{m=1}^n g(X_m;(\mc A\Gamma\upphi|_{\square_{X_m}},0);\upphi|_{X_m}) = \tilde Z_{\geq 4}(\mc A\Gamma\upphi)
\end{align}
for any fixed $\upphi $. Since
\begin{align}
&e^{\frac g4 \bar\phi^4|X| + \int_{\xi\in\square_X} (\hbar + \eps-\bar G^{-1}\bar\phi)(\psi_\xi+\psi'_\xi) +\frac32 g\bar\phi^2(\psi_\xi+\psi'_\xi)^2 - \int_{\xi\in \square_{X\setminus R}} \frac g4 (\psi_\xi+\psi'_\xi+\bar\phi)^4} \nnb & \qquad = e^{\frac g4 \bar\phi^4|X\cap R| + \int_{\xi\in\square_{R}} (\hbar + \eps-\bar G^{-1}\bar\phi)(\psi_\xi+\psi'_\xi) +\frac32 g\bar\phi^2(\psi_\xi+\psi'_\xi)^2 - \int_{\xi\in \square_{X\setminus R}} V_\eps(\psi_\xi+\psi'_\xi)} ,
\end{align}
this is equivalent to
\begin{align}
Z_{\geq 4}(\mc A\Gamma \upphi + \bar\phi ) &= \sum_{\{X_m\}_1^n\in\mc P(\Lambda)}  \prod_{m=1}^n e^{-\int_{\xi\in\square_{X_m}\setminus R(\upphi|_{X})}\frac g4 \psi_\xi^4 } g^{R(\upphi|_{X_m})} (X;\psi)\Big|_{\psi=\mc A\Gamma\upphi  + \bar\phi} \nnb &\qqquad\times \delta_1(\upphi,\{X_m\}_1^n)\times \delta_2(\upphi,\{X_m\}_1^n),\label{eqproofpsr}
\end{align}
where
\begin{align}
\delta_1(\upphi,\{X_m\}_1^n) &= \prod_{m=1}^n\delta(X_m\supset R(\upphi|_{X_m}))\\
\delta_2(\upphi,\{X_m\}_1^n) &= \prod_{m=1}^n\delta(\mc A\Gamma\upphi|_{\square_{X_m}}+\bar\phi\in \mc D(R(\upphi|_{X_m}),X_m ) )
\end{align}
Set $ D=R(\upphi) $. Note that $\delta_1(\cdots)=0 $ unless $D=\cup_m R(\upphi|_{X_m}) $, and therefore $R(\upphi|_{X_m}) = R(\upphi|_{X_m}) \cap X_m = D\cap X_m $, and $\delta_2(\upphi,\{X_m\}_1^n)=1 $ on the support of $\delta_1 $. By the analytic continuation property of $g^D $ (see (\ref{ancont})), which is applicable by (\ref{lfinclproperty}), we conclude that the right hand side of (\ref{eqproofpsr}) equals
\begin{align}
\cdots &= e^{-\int_{\xi\in \square_{D^c}} \frac g4 \psi_\xi^4}\sum_{\{X_m\}_1^n\in\mc P(\Lambda)}  \prod_{m=1}^n g^{D} (X;\psi)  \times \delta\big(X_m\supset \overline{D\cap X_m} \big)\Big|_{\psi=\mc A\Gamma\upphi + \bar\phi}\\\nonumber
&= e^{-\int_{\xi\in \square_{D^c}} \frac g4 \psi_\xi^4}\sum_{\{X_m\}_1^n\in\mc P_D(\Lambda)} \prod_{m=1}^n g^{D} (X;\psi) \Big|_{\psi=\mc A\Gamma\upphi + \bar\phi} = Z_{\geq 4}(\mc A\Gamma\upphi+\bar\phi)
\end{align}
by (\ref{gkpsrep2}), as was to be shown.

\end{proof}

\begin{rem}\label{propPhi}
While $g(X;\uppsi;\upphi) $ is defined for all $\uppsi\in\Psi(X) $, it is set equal to a trivial zero unless $(\psi + \psi')|_{\square_X}+\bar\phi \in \mc D(R(\upphi|_X),X) $. In (\ref{psz}), as well as in various quantities appearing in the cluster expansion of Proposition \ref{propce} below, expressions like $g(X;(\mc A\Gamma\upphi|_{\square_X},0);\upphi|_X) $ appear, and one might wonder whether these are ever evaluated to such trivial values of $g(X;\uppsi;\upphi) $.\\
It is a somewhat subtle point that, while checking whether $\mc A\Gamma\upphi |_{\square_X}+\bar\phi \in \mc D(R(\upphi|_X),X) $ requires the knowledge of $\upphi_x $ also for $x\in X^c$, for products over a partition such as in (\ref{psz}), it is sufficient for each factor $g(X_m;\uppsi;\upphi) $ to check a \emph{local} condition such as $X\supset \overline{R(\upphi|_X)} $. According to (\ref{lfinclproperty}), this will ensure that $g(X;\uppsi;\upphi) $ will never be evaluated trivially as long as it appears in such a product over a partition. \\
Philosophically speaking, this is similar as in a social contract where nobody needs to defend themselves against others because everybody agreed to restrain their own violent tendencies: As long as every activity $g(X;(\mc A\Gamma\upphi|_{\square_X},0);\upphi|_X) $ checks the local condition that any point $x\in X $ at which $\upphi_x $ is large is, in fact, contained deep inside $X$ (namely, $X\supset \overline{R(\upphi|_X)} $), none of these factors needs to check the size of $\upphi_x $ for $x\not\in X $ (a non-local condition), since if $\upphi_x $ was large at a point close to $X$, $X$ would have to overlap with the polymer containing this point, which is a contradiction. This ensures that polymer system representations with large field conditions, such as (\ref{gkpsrep2}), still ``factorize'', i.e. can be written as polymer systems with large field conditions checked locally.

\erem
\end{rem}

\para{Cluster and Mayer expansion for $Z$}{cmexp}
Our next step is to use a cluster expansion that writes $Z$ of (\ref{psz}) as a polymer system with blocks of size $\mc L^4 $, and to compute the logarithm of this polymer system via the Mayer expansion, if the expansion converges. More precisely, defining the set of partitions (polymers) of $\sfX\subset \Uplambda $ as
\begin{align}
\mc P(\sfX) &= \big\{\{\sfX_m\}_1^n,\, \sfX_{m}\cap \sfX_{m'}=\emptyset,\, \cup_m\sfX_m = \sfX\big\}
\end{align}
we have the following Proposition.

\begin{prop}\label{propce}
For any $\mc L\in \ell\cdot\mb N $, we have the cluster expansion
\begin{align}\label{cealg}
Z= \sum_{\{\sfX_m\}_1^n\in\mc P(\Uplambda)}\prod_{m=1}^n \sfA(\sfX_m),
\end{align}
where 
\begin{align}\label{is}
\sfA(\sfX) &= \int A(\sfX;\upphi|_{\square_\sfX})  \ud\mu(\upphi|_{\square_\sfX}),
\end{align}
with the localized activities
\begin{align}\label{loc}
A(\sfX;\upphi|_{\square_\sfX}) &= \sums{\{\sfX_m\}_1^n\in\mc P(\sfX) \\  T\text{ d-tree on }\ul n} \prod_{e\in T}\Big[\int_0^1 \ud s_e\Big] \ \prod_{m=1}^n \Big[\mf D^{\upphi}_{m,T,\{\sfX_m\}} \sfg(\sfX_m;\mc A( \Gamma(\bd s_m\cdot\upphi) ,\phi');\upphi|_{\square_{\sfX_m}}) \Big]_{\phi' =0}\\\label{diffop}
\mf D^{\upphi}_{m,T,\{\sfX_m\}} &= \prod_{m':(m,m')\in T} \Big[ \sums{x\in\Lambda \\  y\in \square_{\sfX_{m'}}} \Gamma_{x,y}\upphi_y\partial_{\phi'_x} \Big]\\
(\bd s_m )_x &= \min\{s_e,\, e\in \text{the }T\text{ path from }m\text{ to }m',\text{ where }x\in\square_{\sfX_{m'}}\}\label{defintpol}
\end{align}
and the reblocked activities
\begin{align}\label{rebloc}
\sfg(\sfX;\uppsi;\upphi ) &= \sums{\{X_m\}_1^n \in \mc P(\square_\sfX) \text{ s.t.}\\ \mc G(\square(\{X_m\}_1^n))\text{ connected}} \prod_{m=1}^n g(X_m;\uppsi;\upphi)\\
\mc G(\square(\{X_m\}_1^n)) &= \text{graph on }\ul n\text{ with: }\{m,m'\}\in \mc G \Leftrightarrow \exists\, \sfx\text{ s.t. }\square_\sfx\cap  X_m ,\square_\sfx\cap  X_{m'} \neq\emptyset  .
\end{align}
$\sfg(\sfX;\uppsi;\upphi) $ depends only on $\uppsi_\xi,\upphi_x $ for $\xi\in\boxdot_\sfX,x\in\square_X $, and satisfies $\partial_{\psi_\xi}\sfg(\sfX;\uppsi;\upphi) = \partial_{\psi'_\xi}\sfg(\sfX;\uppsi;\upphi) $ and similarly for higher derivatives, whenever $\uppsi\in \Psi(\square_\sfX,\upphi) $. Furthermore, we have
\begin{align}\label{mealg}
\log Z &= \sum_{\sfx\in\Uplambda} \log \sfA(\{\sfx\}) + \sum_{n\geq 1}\tfrac1{n!}\sums{\sfX_1,\ldots,\sfX_n\subset\Uplambda\\ |\sfX_m|>1} \rho\big((\sfX_m)_1^n\big) \prod_{m=1}^n \frac{\sfA(\sfX_m)}{\prod_{\sfx\in\sfX_m}\sfA(\{\sfx\})}
\end{align}
whenever $\Re\sfA(\{\sfx\}) >0$ for all $\sfx$ and the series converges absolutely. Here, the Ursell function
\begin{align}
\rho\big((\sfX_m)_1^n\big) &= \sum_{\mc G \text{ connected graph on }\ul n} \prod_{\{m,m'\}\in \mc G} (-\delta_{\sfX_m\cap\sfX_{m'} \neq\emptyset}).
\end{align}
\end{prop}

\vspace{10pt}

\noindent In the Proposition, we used the notion of a d-tree, which is a directed graph that becomes a tree when the orientation of the edges are disregarded. In other words, it is a tree together with an orientation of each edge. 
\begin{proof}

The identity 
\begin{align}
\sum_{\{X_m\}_1^n\in\mc P(\Lambda)}\prod_{m=1}^n g(X_m;\uppsi;\upphi) &= \sum_{\{\sfX_m\}_1^n\in\mc P(\Uplambda)}\prod_{m=1}^n\sfg(\sfX_m;\uppsi;\upphi)
\end{align}
follows immediately by decomposing the connectedness graph $\mc G(\square(\{X_m\}_1^n)) $ of all terms on the left hand side into connected components. The claim that derivatives of $\sfg(\sfX;\uppsi;\upphi)$ w.r.t. $\psi$ equal those w.r.t. $\psi'$ as long as $\uppsi\in\Psi(\square_\sfX,\upphi) $ is inherited from the same property of $g$ (see also Remark \ref{propPhi} for why this holds on all of $\Psi(\square_\sfX,\upphi) $).\\ 
Even though $\sfg(\sfX;\uppsi;\upphi) $ depends only on $\uppsi_\xi $ for $\xi\in\boxdot_\sfX $ and $\upphi_x $ for $x\in\square_\sfX $, $\sfg(\sfX;(\mc A\Gamma\upphi,\psi');\upphi) $ depends on all the $\upphi_x $, and we localize this dependence using the BKAR interpolation formula \cite{abdesselam1995trees}. For any fixed partition $\{\sfX_m\}_1^n\in\mc P(\Uplambda) $, we have, by the BKAR formula
\begin{align}
\prod_{m=1}^n\sfg(\sfX_m;(\mc A\Gamma\upphi,0);\upphi)&=  \prod_{m=1}^n\sfg(\sfX_m;(\mc A\phi,0);\upphi)\big|_{\phi = \Gamma(\bd s_m \cdot \upphi)} \ \bigg|_{ s(\{m,m'\})=1\,\forall\, m,m'\in\ul n}\nnb
&\mquad\mquad= \sum_{F\text{ forest on }\ul n} \prod_{e\in F}\Big[\int_0^1\ud s_e\Big] \prod_{e\in F} \partial_{s(e)} \prod_{m=1}^n\sfg(\sfX_m;(\mc A\phi,0);\upphi)\big|_{\phi = \Gamma(\bd s_m \cdot \upphi)} \ \bigg|_{ s(\{m,m'\} )=s^F(\{m,m'\})}\nonumber
\end{align}
where $(\bd s_m)_x = s(\{m,m'\}) $ if $x\in\square_{\sfX_{m'}} $ and 
\begin{align}
s^F(\{m,m'\}) &= \min\{s_e,\, e\in \text{the }F\text{ path from }m\text{ to }m' \}
\end{align}
Each derivative $\partial_{s(e)} $ can act on precisely two factors in the product over $m $, namely $m=m_1$ or $m=m_2 $ if $e=\{m_1,m_2\} $. If it acts on $m_1 $, we regard the edge $e = (m_1,m_2)$ as directed towards $m_2$, and the other way around. By the product rule, this produces a sum over directed forests:
\begin{align}\label{dforests}
\prod_{m=1}^n\sfg(\sfX_m;(\mc A\Gamma\upphi,0);\upphi)&=\sum_{F\text{ d-forest on }\ul n} \prod_{e\in F}\Big[\int_0^1\ud s_e\Big]  \\\nonumber &\qquad\qquad\times \prod_{e\in F}\partial_{s(e)} \prod_{m=1}^n\sfg(\sfX_m;(\mc A\phi,0);\upphi)\big|_{\phi = \Gamma(\bd s_m \cdot \upphi)} \ \bigg|_{ s(m,m')=s^F(m,m')},
\end{align}
where now $(\bd s_m)_x = s(m,m') $ if $x\in\square_{\sfX_{m'}} $ and 
\begin{align}
s^F(m,m') &= \min\{s_e,\, e\in \text{the }F\text{ path from }m\text{ to }m' \}.
\end{align}
In (\ref{dforests}), every derivative now acts on a unique factor in the product over $m$, and we can write
\begin{align}
\prod_{e\in F} \partial_{s(e)} \prod_{m=1}^n\sfg(\sfX_m;(\mc A\phi,0);\upphi)\big|_{\phi = \Gamma(\bd s_m \cdot \upphi)} &=  \prod_{m=1}^n \Big[\prod_{m':(m,m')\in F}\partial_{s(m,m')}\sfg(\sfX_m;(\mc A\phi,0);\upphi)\big|_{\phi = \Gamma(\bd s_m \cdot \upphi)}\Big].
\end{align}
By the chain rule, and noting, as in Remark \ref{propPhi}, that we are allowed to use the property that derivatives of $\sfg $ w.r.t. $\psi$ equal those w.r.t. $\psi'$, we get
\begin{align}
\partial_{s(m,m')}\sfg(\sfX_m;(\mc A\phi,0);\upphi)\big|_{\phi = \Gamma(\bd s_m \cdot \upphi)} &= \sums{x \in\Lambda\\ y\in\square_{\sfX_{m'}}} \Gamma_{x,y} \upphi_y \partial_{\phi'_x }\sfg(\sfX_m;\mc A( \Gamma(\bd s_m \cdot \upphi),\phi');\upphi)\big|_{\phi'=0}
\end{align}
Inserting, and factoring the sum over forests into a product of sums over trees, we deduce
\begin{align}
\sum_{\{\sfX_m\}_1^n\in\mc P(\Uplambda)}\prod_{m=1}^n\sfg(\sfX_m;(\mc A\Gamma\upphi,0);\upphi) &= \sum_{\{\sfX_m\}_1^n\in\mc P(\Uplambda)}\prod_{m=1}^n A(\sfX_m;\upphi|_{\square_{\sfX_m}}).
\end{align}
Since the measure $\ud\mu(\upphi) $ is a product measure and the activities $ A(\sfX;\upphi|_{\square_\sfX}) $ are localized, this immediately implies the cluster expansion (\ref{cealg}). The Mayer expansion (\ref{mealg}) follows from (\ref{cealg}) by a standard argument similar to the above. See, e.g., \cite{salmhofer1999renormalization}.

\end{proof}

\noindent
All activities are also functions of $\eps  $, but we will usually not display this dependence explicitly.

\section{Bounds on the cluster expansion}\label{bounds}
In this section, we prove bounds on the cluster expansion described in Proposition \ref{propce}. The activity $\sfA(\sfX) $ is built by an integration, a localization, and a reblocking step, and we will establish a bound for each of these steps, using appropriate norms for the intermediate activities $A(\sfX;\upphi)$ and $\sfg(\sfX;\uppsi;\upphi) $.\\ 
We now first define the norms we need. We set
\begin{align}
\Vert \sfA\Vert &:= \sup_{\substack{\sfx\in\Uplambda \\ |\eps|\leq \upepsilon}} \sums{\sfX\ni\sfx\\|\sfX|>1}  (2e)^{|\sfX|-1} |\sfA(\sfX)|\\
\Vert A\Vert &:=\sup_{\substack{\sfx\in\Uplambda \\ |\eps|\leq \upepsilon}}\sums{\sfX\ni\sfx\\|\sfX|>1} (2\,e\,\mc N )^{|\sfX|-1} \sup_{\upphi\in\mb R^{\square_\sfX}} e^{-(\la\frac13\ra + \frac1{10})\Vert\upphi\Vert_\sfX^2} |A(\sfX;\upphi|_{\square_\sfX})|\\
\Vert\sfg\Vert &:= \sup_{\substack{\sfx\in\Uplambda \\ |\eps|\leq \upepsilon}} \sums{\sfX\ni\sfx\\|\sfX|>1} (2\, e^3\,\mc N )^{|\sfX|-1} \sup_{\substack{\uppsi\in\Psi(\square_\sfX,\upphi),\\ \upphi\in\mb R^{\square_\sfX} }} e^{-  \la g\bar\phi^2\ra\Vert\psi\Vert_{L^2(\boxdot_\sfX)}^2 -\frac1{20}\Vert\upphi\Vert_{\sfX}^2  } |\sfg(\sfX;\uppsi;\upphi)|\\
\Vert\sfg\Vert_\square &:= \sup_{\substack{\sfx\in\Uplambda \\ |\eps|\leq \upepsilon}}  \sup_{\substack{\uppsi\in\Psi(\square_\sfx,\upphi),\\\upphi\in\mb R^{\square_\sfx} }} e^{-  \la g\bar\phi^2\ra\Vert\psi\Vert_{L^2(\boxdot_\sfx)}^2 -\frac1{20}\Vert\upphi\Vert_{\{\sfx\}}^2   } |\sfg(\{\sfx\};\uppsi;\upphi)|\\
\Vert g\Vert &:= \sup_{\substack{\Delta\in\mc C \\ |\eps|\leq\upepsilon}} \sums{X\supset\Delta \\|X|_\ell>1}   (2\, e) ^{|X|_\ell-1} \sup_{\substack{\uppsi\in\Psi(X,\upphi),\\ \upphi\in\mb R^X }} e^{-   \la g\bar\phi^2\ra\Vert\psi\Vert_{L^2( \square_ X)}^2 -\frac1{20}\Vert\upphi\Vert_{X}^2  } |g(X;\uppsi;\upphi)|\label{normg}\\
\Vert g\Vert_\Delta &:= \sup_{\substack{\Delta\in\mc C \\ |\eps|\leq\upepsilon}} \sup_{\substack{\uppsi\in\Psi(\Delta,\upphi),\\\upphi\in\mb R^{\Delta} }} e^{-   \la g\bar\phi^2\ra\Vert\psi\Vert_{L^2(  \square_\Delta)}^2 -\frac1{20}\Vert\upphi\Vert_{\Delta}^2 } |g(\Delta;\uppsi;\upphi)|\label{normgdelta}
\end{align} 
$\mc N $ and all parameters are still the same as in section \ref{remheur}, and we used
\begin{align}
\Vert\psi\Vert_{L^2(X_0)}^2 &= \int_{\xi\in X_0}|\psi_\xi|^2 \\\label{defXnorm}
\Vert\upphi\Vert_X &= \mc L^{-4}\sum_{x\in X}\upphi_x^2
\end{align}
Also, we will separately need to study the activity $\sfA(\{\sfx\}) $ on a single point, which is just a complex number and the rigorous analog of the single point activity $\sfB(\square_\sfx) $ of section \ref{sbheur}.\\
The choice of our norms is justified by the fact that, as we shall show, each of them can be controlled naturally in terms of the next one, and that the Mayer expansion satisfies the following well known bound, in terms of $\Vert \sfA\Vert $:

\begin{prop}\label{propboundme}
Let $\log Z $ be defined in terms of $\sfA(\sfX) $ as in (\ref{mealg}). Suppose that $\Re \sfA(\{\sfx\})>\frac12 $, and that $ \Vert\sfA\Vert<  \min_\sfx|\sfA(\{\sfx\})| $ in the norm defined above. Then, we have the bound
\begin{align}\label{boundme}
|\log Z| &\leq \sum_{\sfx\in\Uplambda} |\log \sfA(\{\sfx\})|  + |\Uplambda|\cdot \tfrac12\cdot \sum_{n\geq 1}\tfrac1{n^2}\Big( \tfrac{ \Vert\sfA\Vert}{ \min_\sfx|\sfA(\{\sfx\})|}\Big)^n .
\end{align}
\end{prop}

\noindent
We include a proof of this standard result for the convenience of the reader. 
\begin{proof}
It is a well-known property of the Ursell function that
\begin{align}
|\rho((\sfX_m)_1^n)|\leq \sum_{T\subset \mc G((\sfX_m)_1^n)} \prod_{\{m,m'\}\in T}\delta_{\sfX_m\cap\sfX_{m'}\neq\emptyset}.
\end{align}
Here, $\{m,m'\}\in\mc G((\sfX_m)_1^n)\Leftrightarrow \sfX_m\cap\sfX_{m'}\neq\emptyset $, and the sum is over spanning trees of that graph (in particular, $ \rho((\sfX_m)_1^n)=0$ if the graph is not connected). Inserting this bound and $\sum_{\sfX_1}f(\sfX_1)\leq |\Uplambda|\sup_{\sfx_1}\sum_{\sfX_1\ni\sfx_1}|\sfX_1|^{-1}f(\sfX_1) $ into (\ref{mealg}), we get
\begin{align}
|\log Z| &\leq \sum_{\sfx\in\Uplambda} |\log \sfA(\{\sfx\})|  + |\Uplambda|\cdot \sum_{n\geq 1}\tfrac1{n!}\sum_{T\text{ tree on }\ul n} Z(T)
\end{align}
with
\begin{align}
Z(T) &= \sup_{\sfx_1\in\Uplambda}   \sums{\sfX_1,\ldots,\sfX_n\subset\Uplambda \\  \sfX_1\ni \sfx}Z(T;(\sfX_m)_1^n)\label{treerepbme}\\
Z(T;(\sfX_m)_1^n) &= |\sfX_1|^{-1} \cdot \prod_{\{m,m'\}\in T}\delta_{\sfX_m\cap\sfX_{m'}\neq \emptyset} \cdot \prod_{m=1}^n \frac{|\sfA(\sfX_m)|\cdot \delta_{|\sfX_m|>1}}{\prod_{\sfx\in\sfX_m}|\sfA(\{\sfx\})| }
\end{align}
Consider the case $n\geq 2 $. Regard $1$ as the root of $T$. We now describe a procedure that bounds $Z(T) $ by a product of identical quantities, each associated to a tree rooted at one of the neighbors of $1$, and obtained by deleting the edges incident to $1$. Deleting a root $m$ will produce a factor 
\begin{align}\label{meindfactors}
\sfA_m &= (\tfrac 12)^{\delta_{m>1}}\sup_{\sfx\in\Uplambda} \sum_{\sfX\ni\sfx} |\sfX|^{d^T_m-1}  \frac{|\sfA(\sfX)|\cdot \delta_{|\sfX|>1}}{\prod_{\sfx\in\sfX}|\sfA(\{\sfx\})| } 
\end{align}
where $d^T_m $ is the degree of $m$ in $T$. Applying the same procedure to the factors so obtained, until one reaches the leafs of the tree, will thus give the bound
\begin{align}\label{meindbound}
Z(T)  \leq \prod_{m=1}^n\sfA_m ,
\end{align}
from which the Proposition will follow by a combinatorial argument.\\
Denote the neighbors of $1 $ by $m_1,\ldots,m_d $ with $d= d^T_1 $, and let $T_1,\ldots, T_d $ be the trees of the forest obtained from $T$ by deleting all edges incident to $1$. We denote the vertex set of $T_c$ by $N_c$, and $N_1,\ldots,N_d $ forms a partition of $\{2,\ldots,n\} $. Using
\begin{align}
\delta_{\sfX_1\cap\sfX_{m_c}\neq \emptyset} \leq \frac{|\sfX_1\cap\sfX_{m_c}|}{|\sfX_{m_c}|} =  |\sfX_{m_c}|^{-1}\sum_{\sfx\in\sfX_{m_c}} \delta_{\sfx\in\sfX_1}
\end{align}
and
\begin{align}
\sum_{\sfx\in\Uplambda }\sum_{\sfX_{m_c}\ni \sfx}  \delta_{\sfx\in\sfX_1} \, f(\sfX_{m_c})\leq |\sfX_1| \sup_{\sfx\in\Uplambda} \sum_{\sfX_{m_c}\ni \sfx}f(\sfX_{m_c}),
\end{align}
we bound
\begin{align}
Z(T) &\leq \Big[\sup_{\sfx_1\in\Uplambda} \sum_{\sfX_1\ni\sfx_1} |\sfX_1|^{d-1}\frac{|\sfA(\sfX_1)|\cdot \delta_{|\sfX_1|>1}}{\prod_{\sfx\in\sfX_1}|\sfA(\{\sfx\})| } \Big] \times \prod_{c=1}^d Z(T_c),
\end{align}
with $Z(T_c)$ defined just like $Z(T) $, but for the smaller tree $T_c$, rooted at $m_c$. We may thus proceed to remove roots inductively. Since the coordination number of $m_c $ in $T_c $ is one smaller than the coordination number in $T$, we produce in this way factors like (\ref{meindfactors}), but without the factor $\frac12  $ and $|\sfX|^{d^T_m-1} $ replaced by $ |\sfX|^{d^T_m-2}$. Bounding $|\sfX|^{d^T_m-2}\leq \frac 12|\sfX|^{d^T_m-1}$ then gives (\ref{meindbound}).\\
Since $|\sfX|^{d-1}\leq (d-1)!e^{|\sfX|-1} $, we have 
\begin{align}
\sfA_m\leq (\tfrac12)^{\delta_{m>1}}\cdot (d^T_m-1)!\cdot  \frac{1\,}{2 \min_\sfx|\sfA(\{\sfx\})|}\Vert \sfA\Vert,
\end{align}
and thus get
\begin{align}\label{cayleayinput}
Z(T) &\leq (\tfrac12)^{n-1}\Big( \tfrac{ \Vert\sfA\Vert}{2 \min_\sfx|\sfA(\{\sfx\})|}\Big)^n\prod_{m=1}^n (d^T_m-1)!.
\end{align}
For $n\geq 2$, Cayley gives the formula $(n-2)! \prod_{m=1}^n (d^T_m-1)!^{-1} $ for the number of trees on $\ul n$ with coordination numbers $d^T_1+\cdots+d^T_n = 2(n-1) $, and there are $\binom{2(n-1)-1}{n-1} $ choices for such coordination numbers. Thus 
\begin{align}
\sum_{n\geq 2}\tfrac1{n!}\sum_{T\text{ tree on }\ul n}Z(T)&\leq \sum_{n\geq 2}\tfrac1{n(n-1)} \binom{2(n-1)-1}{n-1}(\tfrac12)^{n-1}\Big( \tfrac{ \Vert\sfA\Vert}{2 \min_\sfx|\sfA(\{\sfx\})|}\Big)^n\nnb &\leq  \tfrac12 \sum_{n\geq 2}\tfrac1{n^2}\Big( \tfrac{  \Vert\sfA\Vert}{ \min_\sfx|\sfA(\{\sfx\})|}\Big)^n
\end{align}
The bound on the $n=1 $ term is trivial. This implies the proposition.

\end{proof}

\para{Bound on the integration step.}{bis} We start our analysis of the cluster expansion by bounding the integration step, which is extremely simple.

\begin{prop}
Let $\sfA $ be defined in terms of $ A$ as in (\ref{is}). Then, we have the bound
\begin{align}\label{boundis}
\Vert \sfA\Vert \leq \mc N\cdot \Vert A\Vert,
\end{align}
where $\mc N $ is as in (\ref{defncal}), with $\alpha = \la\frac13\ra + \frac1{10}<\frac12 $.
\end{prop}

\begin{proof}
This follows immediately from
\begin{align}
|\sfA(\sfX)| & \leq \int e^{(\la\frac13\ra + \frac1{10})\Vert\upphi\Vert_\sfX^2}\ud\mu(\upphi|_{\square_\sfX})\cdot \sup_{\upphi\in\mb R^{\square_\sfX}} e^{-(\la\frac13\ra + \frac1{10})\Vert\upphi\Vert_\sfX^2} |A(\sfX;\upphi|_{\square_\sfX})| \nnb
&= \mc N^{|\sfX|} \cdot \sup_{\upphi\in\mb R^{\square_\sfX}} e^{-(\la\frac13\ra + \frac1{10})\Vert\upphi\Vert_\sfX^2} |A(\sfX;\upphi|_{\square_\sfX})|\nonumber 
\end{align}
\end{proof}

\para{Bound on the localization step.}{bloc} We now prove a bound on the localization step. As in (\ref{entrbound}) for the case of two blocks, the factors of $\Gamma$ in (\ref{diffop}) (and connectedness of trees) will be used to control the diameter of polymers, and as in (\ref{1stboundb2block}), the derivatives in (\ref{diffop}) will be used to control the number of blocks in each polymer. The derivatives will again be estimated using the Cauchy formula, but for large numbers of blocks, factorials will be produced, and we shall need the following combinatorial tool to control these factorials:
\begin{align}\label{treebound}
\sum_{T\text{ d-tree on }\ul n} \prod_{m=1}^n d^T_m! &\leq  (n-2)! \; 2^{4(n-1)}\qquad\forall \, n\geq 2
\end{align}
Above, $d^T_m $ is the degree of $m$ in $T$ (i.e. number of adjacent edges of either orientation). The above bound follows from the same argument as the one after (\ref{cayleayinput}), and is well known.

\begin{prop}\label{proplocbound}
Let $A$ be defined in terms of $\sfg $ as in (\ref{loc}). Then, we have the bound
\begin{align}\label{boundloc}
\Vert A\Vert &\leq  \Vert\sfg\Vert +   \sum_{n= 2}^{|\Uplambda|} \tfrac1{n-1}  \big(c_{loc} \cdot   \mf r^{-1} \cdot \Vert\mc A\Vert\cdot \Vert\Gamma\Vert\big)^{n-1}\cdot  \big(\Vert\sfg\Vert_\square+\Vert\sfg\Vert \big)^{n} ,
\end{align}
where $\Vert \mc A\Vert $ is as in (\ref{normacal}), $\Vert\Gamma\Vert $ as in (\ref{normgamma}), and $c_{loc} $ depends only on $L$. 
\end{prop}

\begin{proof}
We order the sums in (\ref{loc}) as
\begin{align}
A(\sfX;\upphi) &= \sum_{n= 1}^{|\Uplambda|} \frac1{n!} \sum_{T\text{ d-tree on }\ul n}  \prod_{e\in T}\Big[\int_0^1 \ud s_e\Big] A_{T,\bd s}(\sfX;\upphi)
\end{align}
with 
\begin{align}
A_{T,\bd s}(\sfX;\upphi) &= \sum_{(\sfX_m)_1^n\in \mc P_n(\sfX)}\prod_{m=1}^n \Big[\mf D^{\upphi}_{m,T,\{\sfX_m\}} \sfg(\sfX_m;\mc A(\Gamma(\bd s_m\cdot\upphi),\phi');\upphi|_{\square_{\sfX_m}}) \Big]_{\phi' = 0}
\end{align}
Here, we have changed to use ordered partitions of $\sfX\subset\Uplambda$ of length $n$ (there are $n!$ of them for each unordered one), the set of which we denote by $ \mc P_n(\sfX) $. Our goal is to show that  
\begin{align}
\label{treeactiva}
\Vert A_{T,\bd s}\Vert &\leq   n\cdot    \big(\tfrac{c_{loc}}{16} \cdot \mf r^{-1}\cdot \Vert\mc A\Vert\cdot\Vert\Gamma\Vert \big)^{n-1}\cdot \big(\Vert \sfg\Vert_\square+ \Vert\sfg\Vert\big)^n\cdot  \prod_{m=1}^n d^T_m!
\end{align}
for any $n\geq 2$ and $\bd s $ as in (\ref{defintpol}), and that, for $n=1 $ and $T=\emptyset $ (the only tree on one vertex),
\begin{align}\label{1blockbounda}
\Vert A_{\emptyset,-}\Vert &\leq \Vert\sfg\Vert.
\end{align}
We can then invoke (\ref{treebound}) to get the conclusion.\\
We first prove (\ref{1blockbounda}). Since
\begin{align}
A_{\emptyset,-} (\sfX;\upphi) &= \sfg(\sfX;(\mc A\Gamma \upphi|_{\square_\sfX},0);\upphi|_{\square_\sfX}),
\end{align}
(\ref{1blockbounda}) will follow if we can prove 
\begin{align}
\sup_{\upphi\in\mb R^{\square_\sfX}} e^{-(\la\frac13\ra + \frac1{20})\Vert\upphi\Vert_\sfX^2 +   \la g\bar\phi^2\ra \Vert \mc A\Gamma \upphi\Vert_{L^2(\boxdot_\sfX)}^2 }\leq 1,
\end{align}
and this follows immediately from (\ref{schurbound}) and (\ref{sharpboundcov}).\\
To prove (\ref{treeactiva}), we start by inserting (\ref{cvds}) into (\ref{diffop}), and collect the sums over $\sfx,\sfy $ for all edges in $T$. This gives 
\begin{align}
A_{T,\bd s}(\sfX;\upphi) &= \sum_{(\sfX_m)_1^n\in \mc P_n(\sfX)}\sums{\sfx_{e}\in\Uplambda,\, \sfy_e\in \sfX_{m'}\\ \forall \, e=(m,m')\in T}A\big(T;\bd s;\bm\sfx,\bm\sfy;(\sfX_m)_1^n;\upphi\big)
\end{align}
with
\begin{align}
A(\cdots) &= \sums{k_e,j_e\in\ul{\mc L^4}\\\forall\, e\in T }   \prod_{e \in T}\ddot\Gamma_{\sfx_e,j_e;\sfy_e,k_e} \prod_{m=1}^n \bigg[\Big[\prod_{m':e=(m,m')\in T}  \dot\upphi_{\sfy_{e},k_{e}} \partial_{\dot\phi'_{\sfx_{e},j_{e}}}\Big] \sfg(\sfX_m;\mc A(\Gamma(\bd s_m\cdot\upphi),\phi'),\upphi|_{\square_{\sfX_m}}) \bigg]_{\phi' = 0}\nonumber
\end{align}
We will prove below that
\begin{align}\label{locboundcauchy}
(2\,e^3\,\mc N )^{|\sfX|-1}\sup_\upphi e^{-(\la\frac13\ra + \frac1{10})\Vert\upphi\Vert_\sfX^2}\big|A(\cdots)\big|&\leq \big(210\cdot \Vert\mc A\Vert\cdot  \mc N  \cdot \mf r^{-1}\big)^{n-1}\prod_{m=1}^n \Vert\sfg (\sfX_m)\Vert \\  \nonumber &\quad\times  \prod_{e=(m,m')\in T} \Vert \Gamma\Vert_{\sfx_e,\sfy_e} e^{-\frac{c_{\mc A}}2 d(\square_{\sfx_e},\square_{\sfX_m})} \prod_{\sfx\in\Uplambda} \sqrt{\vec d^{\bm \sfy}_\sfx!} \,\ivec d^{\bm \sfx}_\sfx!
\end{align}
where $\Vert \Gamma\Vert_{\sfx,\sfy}$ is as in (\ref{normgamma}) and 
\begin{align}
\Vert\sfg (\sfX)\Vert &= (2\,e^3\,\mc N )^{|\sfX|-1} \sup_{\substack{\uppsi\in\Psi(\square_\sfX,\upphi),\\\upphi\in\mb R^{\square_\sfX}}}  e^{- \la g\bar\phi^2\ra \Vert\psi\Vert_{L^2(\boxdot_\sfX)}^2  - \frac1{20}\Vert\upphi\Vert_\sfX^2 } |\sfg(\sfX;\uppsi;\upphi)|
\end{align}
and we defined the factorials of incoming and outgoing local coordination numbers as 
\begin{align}
\vec d^{\bm\sfy}_\sfx! &= \big|\{e\text{ s.t. }\sfy_e = \sfx  \}  \big|!\\
\ivec d^{\bm \sfx}_\sfx! &= \prod_{m=1}^n \big|\{e=(m,m')\text{ s.t. }\sfx_e=\sfx\}\big|!
\end{align}
We then use the following Lemma, which is also proven below.

\begin{lem}\label{pinchandsum}
Let $n\geq 1$ and $T$ be a d-tree on $\ul n $.  Let $\gamma_{\sfx,\sfy} $ and $\mf d_{\sfx,\sfy} $, $\sfx,\sfy\in\Uplambda$, be nonnegative, and define $\mf d(\sfx,\sfX) =  \max_{\sfy\in\sfX}\mf d_{\sfx,\sfy} $. For numbers $\sfg(\sfX) $, $\sfX\subset \Uplambda $, set 
\begin{align}
a(\sfX) &= \sum_{(\sfX_m)_1^n\in\mc P_n(\sfX)} \sums{\sfx_{e}\in\Uplambda,\, \sfy_e\in \sfX_{m'}\\ \forall \, e=(m,m')\in T} \prod_{m=1}^n \sfg(\sfX_m) \prod_{e=(m,m')\in T}\gamma_{\sfx_e,\sfy_e}   \mf d(\sfx_e,\sfX_m)  \prod_{\sfx\in\Uplambda} \sqrt{\vec d^{\bm \sfy}_\sfx!} \,\ivec d^{\bm \sfx}_\sfx!.
\end{align}
Then,
\begin{align}\label{pinchbound}
\sup_{\sfx\in\Uplambda}\sum_{\sfX\ni\sfx}e^{-2|\sfX|} a(\sfX) &\leq n\cdot \big(2\cdot \Vert\gamma\Vert \cdot  \Vert\mf d\Vert \big)^{n-1} \Vert \sfg\Vert^n\cdot \prod_{m=1}^n d^T_m!
\end{align}
with
\begin{align}
\Vert\gamma\Vert &= \max\Big\{\sup_{\sfx\in\Uplambda}\sum_{\sfy\in\Uplambda}\gamma_{\sfx,\sfy},\sup_{\sfy\in\Uplambda}\sum_{\sfx\in\Uplambda}\gamma_{\sfx,\sfy}\Big\}\\
\Vert\mf d\Vert &= \max\Big\{\sup_{\sfx\in\Uplambda,d\geq 1}\Big[\sum_{\sfy\in\Uplambda}\mf d_{\sfx,\sfy}^d\Big]^{\frac1d},\sup_{\sfy\in\Uplambda,d\geq 1}\Big[\sum_{\sfx\in\Uplambda}\mf d_{\sfx,\sfy}^d\Big]^{\frac1d}\Big\} \\
\Vert\sfg\Vert &= \sup_\sfx\sum_{\sfX\ni\sfx}\sfg(\sfX)  
\end{align}
\erem\end{lem}

\noindent
We use the Lemma with $\sfg(\sfX) \to\Vert \sfg(\sfX)\Vert $, $\gamma_{\sfx,\sfy} = \Vert \Gamma\Vert_{\sfx,\sfy} $ and  $\mf d_{\sfx,\sfy} = e^{-\frac 12c_{\mc A} d(\square_\sfx,\square_\sfy)} $, which satisfy $\Vert\gamma\Vert = \Vert \Gamma\Vert $ and $\Vert\mf d\Vert\leq  $ const, for a constant depending only on $L$. The Lemma directly gives (\ref{treeactiva}), which concludes the proof of Proposition \ref{proplocbound}.

\end{proof}

\begin{proof}[Proof of (\ref{locboundcauchy})]
Let us denote the left hand side of (\ref{locboundcauchy}) by $S $. We have
\begin{align}
S&\leq (2\, e^3\mc N)^{n-1}\times \prod_{e=(m,m')\in T}\Vert \Gamma\Vert_{\sfx_e,\sfy_e}e^{-\frac{c_{\mc A}}2 d(\square_{\sfx_e},\square_{\sfX_m})}\times \Phi\times \prod_{m=1}^n\sfg_m
\end{align}
where
\begin{align}
\Phi &= \sup_\upphi e^{-\frac1{20}\Vert\upphi\Vert_\sfX^2} \prod_{e\in T}\Big[\sum_{k\in\ul{\mc L^4}} \mc L^{-2\delta_{k>1} + 2\delta_{k=1}} |\dot\upphi_{\sfy_e,k}|\Big]\\
\sfg_m &= (2\,e^3\mc N)^{|\sfX_m|-1}\sup_{j_e\in\ul{\mc L^4}} \sup_\upphi e^{-(\la\frac13\ra + \frac1{20})\Vert\upphi\Vert_{\sfX_m}^2} \Bigg|  \prod_{m':e=(m,m')\in T} \nnb &\qqquad  \mc L^{-4\delta_{j_e>1}}e^{\frac{c_{\mc A}}2 d(\square_{\sfx_e},\square_{\sfX_m})} \partial_{\dot\phi'_{\sfx_e,j_e}} \sfg(\sfX_m;\mc A(\Gamma(\bd s_m\cdot\upphi),\phi');\upphi)\;\Bigg|_{\phi' = 0}
\end{align}
To bound $ \Phi$, note that
\begin{align}
\sum_{k\in\ul{\mc L^4}} \mc L^{-2\delta_{k>1} + 2\delta_{k=1}} |\dot\upphi_{\sfy,k}|&\leq \prod_{k\in\ul{\mc L^4}}\big(1+ \mc L^{-2\delta_{k>1} + 2\delta_{k=1}} |\dot\upphi_{\sfy,k}|\big)\nnb &\leq \sqrt{20d_\sfy}\;e^{\frac1{20d_\sfy}\Vert \upphi\Vert_{\{\sfy\}}^2 }
\end{align}
since $1+x\leq \sqrt d e^{-\frac1d x^2} $. Choosing $d_\sfy = |\{e\text{ s.t. }\sfy_e=\sfy\}| $, and using $\sum_\sfy d_\sfy =  n-1  $, we get
\begin{align}
\Phi &\leq \prod_{\sfy\in\Uplambda}\sqrt{20 d_\sfy}^{d_\sfy}\leq  (20\cdot e)^{\frac{n-1}2}\prod_{\sfy\in\Uplambda} \sqrt{\vec d^{\bm\sfy}_\sfy!}.
\end{align}
For $\sfg_m $, we use the Cauchy formula to bound
\begin{align}
\Big|\prod_{\sfx,j}\partial_{\dot\phi'_{\sfx,j}}^{d_{\sfx,j}}  \sfg(\sfX;\mc A(\phi, \phi' );\upphi) \Big|_{\phi'=0}&\leq \prod_{\sfx,j}\big[d_{\sfx,j}!\cdot \big(\mc L^{4\delta_{j>1}}\cdot \mf r^{-1} \cdot \Vert \mc A\Vert \cdot e^{-\half d(\square_\sfx,\square_\sfX) } \big)^{d_{\sfx,j}} \big] \nnb &  \qquad \times \sup_{|\dot\eta_{\sfx,j}|\leq  \mc L^{-4\delta_{j>1}} \mf r\,\Vert\mc A\Vert^{-1} \, e^{\frac{c_{\mc A}}2 d(\square_\sfx,\square_\sfX) }} \big|\sfg(\sfX;\mc A(\phi,\eta);\upphi)\big|
\end{align}
By (\ref{cauchycv}), the field $\eta$ satisfies $|\eta_x| \leq 2 \mf r\,\Vert \mc A\Vert^{-1} e^{\frac{c_{\mc A}}2 d(\square_\sfx,\square_\sfX)} $ for all $x\in\square_\sfx $ and all $\sfx\in\Uplambda $, and thus $\mc A\eta\in\mc K_{\mf r}(\square_\sfX) $. Using these bounds with $d_{\sfx,j} = |\{e=(m,m')\text{ s.t. } (\sfx_e,j_e)=(\sfx,j)\}|$, which satisfies $ \sum_jd_{\sfx,j} = d^{\bm\sfx_m}_\sfx$ with $d^{\bm\sfx_m}_\sfx:= |\{e=(m,m')\text{ s.t. } \sfx=\sfx\}| $, we conclude
\begin{align}
\sfg_m &\leq  (e^2\mc N\,K)^{|\sfX_m|-1}\prod_{\sfx\in\Uplambda}d^{\bm\sfx_m}_\sfx!\cdot (\mf r^{-1}\cdot\Vert\mc A\Vert)^{d^{\bm\sfx_m}_\sfx} \\\nonumber  &\qquad\qquad\qquad \times  \sup_{\upphi}e^{- (\la\frac13\ra + \frac1{20})   \Vert \upphi\Vert_{\sfX_m}^2} \sup_{\psi'\in\mc K_{\mf r}(\square_{\sfX_m})} \big|\sfg(\sfX_m;(\mc A\Gamma(\bd s_m\cdot\upphi) , \psi');\upphi)\big|
\end{align}
The result now follows from $\mc A\Gamma(\bd s_m\cdot\upphi)|_{\boxdot_\sfX}\in \Psi(\square_\sfX,\upphi)  $ (unless $S=\prod_m\sfg_m = 0$, as explained in Remark \ref{propPhi}) and
\begin{align}
\la1\ra \Vert \upphi\Vert_{\sfX}^2\geq   \la 3g\bar\phi^2\ra\Vert\psi\Vert^2_{L^2(\boxdot_\sfX)}\big|_{\psi=\mc A\Gamma(\bd s\cdot\upphi)}  ,
\end{align}
which is proven in the same way as (\ref{stabboundfields}).

\end{proof}

\begin{proof}[Proof of Lemma \ref{pinchandsum}]

Let $S$ denote the left hand side of (\ref{pinchbound}). By symmetry,
\begin{align}\label{pinchboundn}
S= n\sup_{\sfx_1\in\Uplambda} \sums{\sfX_1,\ldots,\sfX_n\subset\Uplambda\\\sfX_m\cap\sfX_{m'}=\emptyset\\ \sfX_1\ni\sfx_1} S\big((\sfX_m)_1^n\big),
\end{align}
where
\begin{align}
S\big((\sfX_m)_1^n\big) &= \sums{\sfx_{e}\in\Uplambda,\, \sfy_e\in \sfX_{m'}\\ \forall \, e=(m,m')\in T} \prod_{m=1}^n e^{-2|\sfX_m|}\sfg(\sfX_m) \prod_{e=(m,m')\in T}\gamma_{\sfx_e,\sfy_e}   \mf d(\sfx_e,\sfX_m)  \prod_{\sfx\in\Uplambda} \sqrt{\vec d^{\bm \sfy}_\sfx!} \,\ivec d^{\bm \sfx}_\sfx!.
\end{align}
Regard $1$ as the root of (the undirected version of) $T$. We now describe a procedure, similar to the one of the proof of Proposition \ref{propboundme}, that bounds $S$ by a product of almost identical quantities, each associated to a tree rooted at one of the neighbors of $1$, and obtained by deleting the edges incident to $1$. Deleting an edge produces a factor of $\Vert\gamma\Vert\cdot\Vert\mf d\Vert $, while deleting a root produces the factor $d^T_m !\cdot 2^{d^T_m-1}\cdot\Vert\sfg\Vert $. The same procedure can then be applied to the factors so obtained, until one reaches the leafs of the tree. The bound (\ref{pinchbound}) is simply the product of all factors produced by the iteration of this procedure.\\
Let $\vec d $ (resp. $\ivec d $) be the number of outgoing (resp. incoming) edges at $1$, so that $\vec d+ \ivec d = d^T_1 $, let $\vec m_c ,\, c = 1,\ldots, \vec d $ (resp. $\ivec m_c,c=1,\ldots,\ivec d $) be the vertices adjacent to $1$ via the outgoing (resp. incoming) edges. Deleting the edges adjacent to $1$ produces a forest of $ d^T_1$ d-trees, denoted $T_{\vec m_c}$ and $T_{\ivec m_c} $, and we regard them as rooted at the vertices $\vec m_c,\ivec m_c  $, and denote the connected component of $\{2,\ldots,n\}$ containing $\vec m_c $ (resp. $\ivec m_c $) as $\vec N_c $ (resp. $\ivec N_c $). They form a partition of $\{2,\ldots,n\} $.\\
With these definitions, we have
\begin{align}
S\big((\sfX_m)_1^n\big) &= e^{-2|\sfX_1|} \sfg(\sfX_1) \sums{\vec\sfx_c\in\Uplambda,\vec\sfy_{c}\in \sfX_{\vec m_c}\\ c = 1,\ldots,\vec d} \prod_{c=1}^{\vec d} \gamma_{\vec\sfx_c,\vec\sfy_c} \mf d(\vec\sfx_c,\sfX_1) S((\sfX_m)_{m\in \vec N_c};\vec\sfy_c) \cdot \prod_{\sfx\in\Uplambda} d_\sfx^{\vec{\bm\sfx}}!  \\\nonumber
&\qquad\qquad\qquad\times \sums{\ivec\sfx_c\in\Uplambda,\ivec\sfy_{c}\in \sfX_{1}\\ c = 1,\ldots,\ivec d} \prod_{c=1}^{\ivec d} \gamma_{\ivec\sfx_c,\ivec\sfy_c} \mf d(\ivec\sfx_c,\sfX_{\ivec m_c})S((\sfX_m)_{m\in \ivec N_c};\ivec \sfx_c)\cdot \prod_{\sfy\in\sfX_1} \sqrt{d^{\ivec{\bm\sfy} }_\sfy! },
\end{align}
where 
\begin{align}\label{lempinchindidentity}
S((\sfX_m)_{m\in \vec N_c};\vec\sfy_c) &=  \sums{\sfx_{e}\in\Uplambda,\, \sfy_e\in \sfX_{m'}\\ \forall \, e=(m,m')\in T_{\vec m_c}} \prod_{m\in \vec N_c} e^{-2|\sfX_m|}\sfg(\sfX_m) \prod_{e=(m,m')\in T_{\vec m_c}}\gamma_{\sfx_e,\sfy_e}   \mf d(\sfx_e,\sfX_m)  \prod_{\sfx\in\Uplambda} \sqrt{\vec d^{\bm \sfy,\vec\sfy_c}_\sfx!} \,\ivec d^{\bm \sfx}_\sfx! \nnb\nonumber 
S((\sfX_m)_{m\in \ivec N_c};\ivec\sfx_c)&= \sums{\sfx_{e}\in\Uplambda,\, \sfy_e\in \sfX_{m'}\\ \forall \, e=(m,m')\in T_{\ivec m_c}} \prod_{m\in \ivec N_c} e^{-2|\sfX_m|}\sfg(\sfX_m) \prod_{e=(m,m')\in T_{\ivec m_c}}\gamma_{\sfx_e,\sfy_e}   \mf d(\sfx_e,\sfX_m)  \prod_{\sfx\in\Uplambda} \sqrt{\vec d^{\bm \sfy}_\sfx!} \,\ivec d^{\bm \sfx,\ivec\sfx_c}_\sfx! 
\end{align}
and the factorials are given by
\begin{align}
d^{\vec{\bm\sfx}}_\sfx! &= \big|\{c\text{ s.t. }\vec\sfx_c = \sfx\}\big|!\\
d^{\ivec{\bm\sfy}}_\sfy! &= \big|\{c\text{ s.t. }\ivec\sfy_c = \sfy\}\big|!\\
\vec d^{\bm \sfy,\vec\sfy_c}_\sfx! &= (\big|\{e\text{ s.t. }\sfy_e = \sfx  \}  \big| + \delta_{\sfx,\vec\sfy_c})! \label{newfactorial1}\\ 
\ivec d^{\bm \sfx,\ivec\sfx_c}_\sfx! &=(\big|\{e=(\ivec m_c,m')\text{ s.t. }\sfx_e=\sfx\}\big|+\delta_{\sfx,\ivec\sfx_c})!  \times \prod_{\substack{m\in \ivec N_c \\ m\neq \ivec m_c  } } \big|\{e=(m,m')\text{ s.t. }\sfx_e=\sfx\}\big|! \,.\label{newfactorial2}
\end{align}
(The other factorials in $S((\sfX_m)_{m\in \vec N_c};\vec\sfy_c),S((\sfX_m)_{m\in \ivec N_c};\ivec\sfx_c) $ have the same definitions as before, but restricted to the respective subtrees). We will show below that, for $d\geq 1$,
\begin{align}\label{localfactorials1}
\sums{\sfx_c,\sfy_c\in\Uplambda \\  c=1,\ldots,d} \prod_{c=1}^d \gamma_{\sfx_c,\sfy_c} \mf d(\sfx_c,\sfX) \prod_{\sfx\in\Uplambda} d^{\bm\sfx}_\sfx! &\leq \tfrac 12 \big(2\cdot \Vert \gamma\Vert\cdot \Vert \mf d\Vert\big)^d \,d!\, e^{|\sfX|}\\\label{localfactorials2}
\sums{\sfx_c\in\Uplambda,\sfy_c\in\sfX \\  c=1,\ldots,d} \prod_{c=1}^d \gamma_{\sfx_c,\sfy_c}   \prod_{\sfy\in\sfX} \sqrt{d^{\bm\sfy}_\sfy!}&\leq \tfrac 12 \big(2 \cdot \Vert \gamma\Vert\big)^d \,d!\, e^{|\sfX|}
\end{align}
We neglect the disjointness condition between sets $\sfX_m $ in different connected components, and get
\begin{align}
S&\leq n\cdot \Big(\sup_\sfx\sum_{\sfX_1\ni\sfx} \sfg(\sfX)\Big) \cdot d^T_1!  \cdot  2^{d^T_1-1} \cdot  \Vert\gamma\Vert^{d^T_1} \cdot \Vert \mf d\Vert ^{d^T_1} \times \prod_{c=1}^{\vec d}  S_{\vec N_c}  \times \prod_{c=1}^{\ivec d}   S_{\ivec N_c} 
\end{align}
with
\begin{align}
S_{\vec N_c}  &=  \sup_{\vec \sfy_c}\sums{\sfX_m\subset\Uplambda,m\in\vec N_c\\ \sfX_m\cap\sfX_{m'} = \emptyset \\ \sfX_{\vec m_c}\ni \vec\sfy_c}S((\sfX_m)_{m\in \vec N_c};\vec\sfy_c)
\end{align}
and an identical definition for $S_{\ivec N_c}  $. $S_{\vec N_c}  $ and $S_{\ivec N_c}  $ are very similar to the right hand side of (\ref{pinchboundn}), without the factor of $n$, and for smaller trees (i.e. the subtrees that arise from deleting the edges incident to $1$), except that the factorial factors associated to the root $\vec m_c$ or $\ivec m_c $ have changed, according to (\ref{newfactorial1}) or (\ref{newfactorial2}). But this change just compensates the fact that coordination number of the root in the subtree is one less than the coordination number in $T$ itself (since the edge $(1,\vec m_c) $ or $(\ivec m_c,1) $ has been deleted). We can thus continue inductively to bound $S_{\vec N_c}  $ and $S_{\ivec N_c}  $ in the same way as above. Each edge deleted generates a factor of $\Vert\gamma\Vert \cdot \Vert \mf d\Vert $, while each vertex (root) $m$ processed generates a factor $d^T_m!\cdot 2^{d^T_m-1}\cdot \Vert\sfg\Vert  $. The latter statement holds also if $m$ is a leaf of $T$, since in this case we have $S_{\{\vec m\}} = S_{\{\ivec m\}} = \Vert\sfg\Vert $.  Using $\sum_m d^T_m = 2(n-1) $, we conclude (\ref{pinchbound})

\end{proof}

\begin{proof}[Proof of (\ref{localfactorials1}) and (\ref{localfactorials2})]
(\ref{localfactorials2}) is a consequence of (\ref{localfactorials1}), and so we only prove the latter. We call its left hand side $S$ and order the sum according to coinciding points $\sfx_c$:
\begin{align}
S &= \sum_{n=1}^d\frac1{n!}\sum_{D_1\dot\cup \cdots\dot\cup D_n = \ul d} \sums{\sfx_m\in\Uplambda\\ m\in\ul n  } \prod_{m=1}^n\big(\gamma_{\sfx_m}\mf d(\sfx_m,\sfX)\big)^{|D_m|}\cdot |D_m|!
\end{align}
where
\begin{align}
\gamma_\sfx &= \sum_{\sfy\in\Uplambda} \gamma_{\sfx,\sfy} .
\end{align}
Therefore
\begin{align}
S&\leq \Vert \gamma\Vert^d d!\sum_{n=1}^d\frac1{n!}\sum_{d_1+\cdots+d_n=d}\prod_{m=1}^n \sums{\sfx_m\in\Uplambda\\ m\in\ul n  }  \mf d(\sfx_m,\sfX)^{d_m} \nnb\nonumber
&\leq\big(\Vert \mf d\Vert\cdot \Vert \gamma\Vert\big)^d d!\sum_{n=1}^d\frac{|\sfX|^n} {n!}\binom{d-1}{n-1} \leq \tfrac 12 \big(2\cdot \Vert \gamma\Vert\cdot \Vert \mf d\Vert\big)^d \,d!\, e^{|\sfX|}
\end{align}
as claimed.

\end{proof}

\para{Bound on the reblocking step.}{brebl} We now analyze the reblocking step. We have the following Proposition:

\begin{prop}\label{boundreblock}

Let $\sfg $ be defined in terms of $g $ as in (\ref{rebloc}). Then, if 
\begin{align}\label{singleblockbounda}
\Vert g\Vert_\Delta^{(\frac{\mc L}\ell)^{4}} \leq 2 \cdot (2\, e^3\mc N)^{- 2(\frac{\mc L}\ell)^{-4}} ,
\end{align}
we have the bounds
\begin{align}\label{boundrebl1p}
\Vert \sfg\Vert_\square &\leq \Vert g\Vert_\Delta^{(\frac{\mc L}\ell)^{4}} \sum_{n= 0}^{\frac12(\frac{\mc L}\ell)^4}\tfrac1{n+1}\Big(8\cdot (\tfrac{\mc L}\ell)^4\cdot   \Vert g\Vert\Big)^n
\end{align}
and
\begin{align}\label{boundrebl}
\Vert\sfg\Vert \leq  \sum_{n= 1}^{\frac12|\Lambda|_\ell}\tfrac1{n}\Big(8\cdot (\tfrac{\mc L}\ell)^4\cdot   \Vert g\Vert\Big)^n.
\end{align}

\end{prop}

\begin{proof}

We prove only (\ref{boundrebl}), as (\ref{boundrebl1p}) is much simpler (see also Proposition \ref{bound1p} below). Any partition appearing (\ref{rebloc}) can be thought to consist of a collection of sets $ (X_m)$ with $|X_m|_\ell>1 $ such that $\mc G(\square((X_m))) $ is connected, and a number of blocks $\Delta\in\mc C $ that fill up $\cup X_m $ to become a set of the form $\square_\sfX $. Therefore, changing to ordered partitions and defining
\begin{align}
\mc P_c  &:= \big\{ (X_m)_1^n,\, X_m\cap X_{m'}=\emptyset, \mc G(\square(X_m)_1^n)\text{ connected} , |X_m|_\ell>1\big\}\nnb\nonumber
\mc C((X_m)_1^n) &:= \{\Delta\in\mc C,\,\Delta\subset \square (\cup_m X_m)\setminus \cup_m X_m\},
\end{align}
where $\square(X) =  $ smallest set $\sfX $ with $X\subset\square_\sfX $, we have
\begin{align} 
\Vert\sfg\Vert &\leq   K^{-1}\sum_{n= 1}^{|\Lambda|_\ell/2}\tfrac1{(n-1)!}\sup_{\sfx\in\Uplambda} \sums{(X_m)_1^n\in\mc P_c \\ X_1\cap \square_\sfx \neq\emptyset }\prod_{m=1}^n g(X_m)\times \prod_{\Delta\in \mc C((X_m)_1^n)}g(\Delta)
\end{align}
with $  K = 2\, e^3\mc N  $ and 
\begin{align}
g(X) &=   K^{(\frac{\mc L}\ell)^{-4}|X|_\ell}\sup_{\substack{\uppsi\in\Psi(X,\upphi),\\\upphi\in\mb R^X }}e^{-   \la g\bar\phi^2\ra\Vert\psi\Vert_{L^2( \square_ X)}^2 - \frac1{20}\Vert\upphi\Vert_X^2  } |g(X;\uppsi;\upphi)|.
\end{align}
We used here that $\prod g(X_m;\uppsi;\upphi)\neq 0 $ implies $\uppsi|_{X_m} \in\Psi(X_m,\upphi) $ for all $m$ (see Remark \ref{propPhi}), and that $\Vert\upphi\Vert_\sfX \geq \Vert \upphi\Vert_{\square_\sfX} $ (see (\ref{l2dots}), (\ref{defsfXnorm}) and (\ref{defXnorm})). Overcounting the filling cubes using the inclusion $\mc C((X_m)_1^n)\subset \cup_m\mc C(X_m) $, and discarding disjointness constraints and the explicit factor $  K^{-1} $, may thus bound 
\begin{align}
\Vert\sfg\Vert & \leq    \sum_{n= 1}^{|\Lambda|_\ell/2}\tfrac1{(n-1)!}\sum_{T\text{ tree on }\ul n} g(T)
\end{align}
with
\begin{align}
g(T) &= \sup_{\sfx\in\Uplambda} \sums{X_1,\ldots,X_m\subset \Lambda   \\ X_1\cap \square_\sfx \neq\emptyset }\prod_{m=1}^n g'(X_m) \prod_{\{m,m'\}\in T}\delta_{\square( X_m)\cap\square( X_{m'})\neq\emptyset}  \\
g'(X) &= \delta_{|X|>1}  K^{(\frac{\mc L}\ell)^{-4}|X|_\ell} \Vert g\Vert_\Delta^{|\square(X)\setminus X|_\ell} \sup_{\substack{\uppsi\in\Psi(X,\upphi),\\ \upphi\in\mb R^X }}e^{-   \la g\bar\phi^2\ra\Vert\psi\Vert_{L^2( \square_ X)}^2 - \frac1{20}\Vert\upphi\Vert_X^2 } |g(X;\uppsi;\upphi)|
\end{align}
This is essentially of the same form as (\ref{treerepbme}) in the proof of the bound on the Mayer expansion, and we proceed as in that proof. We regard $1$ as the root of $T$, denote its neighbors by $m_1,\ldots,m_d $, $ d= d^T_1$, which we regard as the roots of the subtrees $T_1,\ldots, T_d $ obtained by deleting edges incident to $1$. Bounding
\begin{align}
\delta_{\square(X_1)\cap\square(X_{m_c})}\leq |\square(X_{m_c})|^{-1} \sum_{\sfx\in \square(\sfX_{m_c})}\delta_{\sfx\in\square(X_1)} 
\end{align}
we get
\begin{align}
g(T) \leq \Big[\sup_{\sfx\in\Uplambda}\sum_{X_1\cap \square_\sfx\neq\emptyset} |\square(X_1)|^{d} g'(X)  \Big]\times\prod_{c=1}^d g(T_c)
\end{align}
with $g(T_c) $ defined just like $g(T) $, but for the smaller tree $T_c$, rooted at $m_c$, and with an extra factor $|\square(X_{m_c})|^{-1} $. Note that
\begin{align}
\sup_{\sfx\in\Uplambda}\sum_{X_1\cap \square_\sfx\neq\emptyset} |\square(X_1)|^{d} g'(X) & \leq d! \,(\tfrac{\mc L}\ell)^4 \sup_{\Delta\in\mc C}\sum_{X\supset\Delta} e^{|X|_\ell-1}g'(X)\nnb
&\leq d! \,(\tfrac{\mc L}\ell)^4\cdot  \tfrac{ K^{2(\mc L/\ell)^{-4}}\, \Vert g\Vert_\Delta^{(\mc L/\ell)^{4} }}{2}\cdot \Vert g\Vert.
\end{align}
by (\ref{singleblockbounda}). Repeating the process and using a bound like (\ref{treebound}), we get the claim of the Proposition.

\end{proof}

\para{Bound on the activity at a single point.}{bsp} Finally, we derive a bound on the activity $\sfA(\{\sfx\}) $ at a single point. By Proposition \ref{propce}, it is given by
\begin{align}\label{1ptactexpl}
\sfA(\{\sfx\}) &=\int  \Big[\sum_{\{X_m\}_1^n\in\mc P(\square_\sfx)}\prod_{m=1}^n g(X_m;((\mc A\Gamma\upphi|_{\square_\sfx})|_{ \square_{X_m}},0);\upphi|_{X_m} )\Big]\,\ud\mu(\upphi|_{\square_\sfx}).
\end{align}
While the proofs of the previous Propositions in principle provide upper bounds on this quantity, we will need a more precise analysis which is suitable to also study $ \log \sfA(\{\sfx\}) $ and $\sfA(\{\sfx\})^{-1} $, as needed in the Mayer expansion. We will achieve this by showing that $\sfA(\{\sfx\})$ is close to the quantity $\sfB(\square_\sfx) $, as was assumed in the heuristics.

\begin{prop}\label{bound1p}

Let $\sfA(\{\sfx\}) $ be defined in terms of $g$ as in (\ref{1ptactexpl}). Then, we have the bound
\begin{align}\label{ebound1p}
\sup_{|\eps|\leq\upepsilon} |\sfA(\{\sfx\})-\sfB(\square_\sfx)| &\leq \la1\ra \sum_{n=1}^{(\frac{\mc L}\ell)^4}\big((\tfrac{\mc L}\ell)^4 \Vert \tilde g\Vert\big)^n
\end{align}
where 
\begin{align}\label{defgtilde}
\tilde g(X;\uppsi;\upphi) &= g(X;\uppsi;\upphi)-\delta_{|X|_\ell=1}e^{-\int_{\xi\in \square_X}V_\eps(\psi_\xi)}
\end{align}
and
\begin{align}\label{normgtilde}
\Vert \tilde g\Vert &:= \sup_{\substack {\sfx\in\Uplambda \\ |\eps|\leq\upepsilon}}\sup_{\emptyset\neq X\subset\square_\sfx} \Bigg[   \sup_{\substack{\uppsi\in\Psi(X,\upphi) \\ \upphi\in\mb R^{X}}} e^{-  \la g\bar\phi^2\ra\Vert\psi\Vert_{L^2(  \square_X)}^2 - \frac1{20}\Vert\upphi\Vert_X^2 } |\tilde g(X;\uppsi;\upphi)|\Bigg]^{1/|X|_\ell}
\end{align}

\end{prop}

\noindent
Note that $\Vert\tilde g\Vert $ also tests the size of $\tilde g $ on a single block, as opposed to $\Vert g\Vert $. Indeed, subtracting $\delta_{|X|_\ell=1}e^{-\int V_\eps} $ from $g$ corresponds to ``first order perturbation theory'', and while $g $ is not small on a single block, $\tilde g $ will be.

\begin{proof}

Since 
\begin{align}
\sum_{\{X_m\}_1^n\in\mc P(X)}\prod_{m=1}^n\delta_{|X_m|_\ell=1} = 1,
\end{align}
for $X$ any paved set, we have
\begin{align}
\sum_{\{X_m\}_1^n\in\mc P(\square_\sfx)}\prod_{m=1}^n g(X_m;\uppsi;\upphi) &= e^{-\int_{\xi\in \boxdot_\sfx}V_\eps(\psi_\xi)}+ \sum_{\emptyset\neq X\subset \square_\sfx }\sum_{\{X_m\}_1^n\in\mc P(X)}\prod_{m=1}^n \tilde g(X_m;\uppsi;\upphi)
\end{align}
and thus
\begin{align}
\sfA(\{\sfx\})-\sfB(\square_\sfx) &=\int \Big[ \sums{\{X_m\}_1^n\in\mc P(X)\\ \emptyset\neq X\subset\square_\sfx}\prod_{m=1}^n \tilde g(X_m;((\mc A\Gamma\upphi|_{\square_\sfx})|_{ \square_{X_m}} ,0);\upphi)\Big]\,\ud\mu(\upphi|_{\square_\sfx}).
\end{align}
We bound
\begin{align}
|\tilde g(X;((\mc A\Gamma\upphi|_{\square_\sfx})|_{ \square_X},0);\upphi )|\leq  \Vert \tilde g\Vert ^{|X|_\ell }\times e^{ \la g\bar\phi^2\ra\Vert \mc A\Gamma\upphi|_{\square_\sfx}\Vert_{L^2(  \square_X)}^2 + \frac1{20}\Vert\upphi\Vert_X^2} .
\end{align}
Similarly to the argument in section (\ref{sbheur}), under the condition (\ref{1blockcond1}), we have
\begin{align}
\int e^{ \la g\bar\phi^2\ra\Vert \mc A\Gamma\upphi|_{\square_\sfx}\Vert_{L^2(  \square_X)}^2+ \frac1{20}\Vert\upphi\Vert_X^2 }  \ud\mu(\upphi|_{\square_\sfx})\leq  \la1\ra
\end{align}
for all $X\subset\square_\sfx $, and so, with $B_n $ the number of partitions of a set of size $n$,
\begin{align}
|\sfA(\{\sfx\})-\sfB(\square_\sfx)| &\leq  \la1\ra\cdot \sum_{\emptyset\neq X\subset\square_\sfx}\Vert \tilde g\Vert^{|X|_\ell} B_{|X|_\ell}\nnb
&\leq  \la1\ra \sum_{n=1}^{(\frac{\mc L}\ell)^4} \frac{((\tfrac{\mc L}\ell)^4 \Vert \tilde g\Vert)^n}{n!}B_n\leq   \la1\ra\sum_{n=1}^{(\frac{\mc L}\ell)^4}\big((\tfrac{\mc L}\ell)^4\Vert \tilde g\Vert\big)^n,
\end{align}
since $B_n\leq n! $. This proves the Proposition.

\end{proof}

\section{Proof of the main theorem}\label{pmt}

In the previous section, we established abstract bounds on the cluster expansion. For them to be useful for our goal, we need to estimate the output of \gk in the norms in which these bounds are formulated. We will see that, if all parameters are chosen appropriately, we can prove estimates that imply Theorem \ref{mainthm} via the machinery of section \ref{bounds}.\\

\para{The result of \gk as an input for the cluster expansion.}{gkce} We now prove bounds on $\Vert g\Vert $ of (\ref{normg}), $\Vert g\Vert_\Delta $ of (\ref{normgdelta}), and $\Vert \tilde g\Vert $ of (\ref{normgtilde}), with $g$ defined in terms of the output of \gk as in (\ref{defg}).\\
We start with the following Lemma, which is the essential estimate we need:

\begin{lem}\label{stabboundact}
Let $g$ be defined as in (\ref{defg}), where $g^D $ is as in Theorem \ref{thmgk}. Define
\begin{align}
\Vert g(X)\Vert  &= \sup_{\substack{\uppsi\in\Psi(X,\upphi),\\\upphi\in\mb R^X }} e^{-  \la g\bar\phi^2\ra\Vert\psi\Vert_{L^2( \square_ X)}^2- \frac1{20}\Vert\upphi\Vert_X^2 } |g(X;\uppsi;\upphi)|.
\end{align}
Then, under the conditions (\ref{sfcondshift}), (\ref{stabcond}), (\ref{1blockcond2}) and (\ref{condconvce}), we have
\begin{align}\label{gbound}
\Vert g(X)\Vert  &\leq e^{3|X|\la \mf r(\upepsilon + 12g\bar\phi\,\mf r^2) \ra  - c_g\mc L (X)}.
\end{align}
Further, let $\tilde g $ be defined as in (\ref{defgtilde}). Then
\begin{align}
\Vert \tilde g(\Delta)\Vert &\leq  e^{  -c(\ell)\,  c_2^2} ,
\end{align}
where $c(\ell)>0 $ depends only on $\ell$. Here, $c_g,c_2$ are the constants from Theorem \ref{thmgk}
\end{lem}

\begin{proof}

Note that (\ref{stabcond}), (\ref{sfcondshift}) and the definition (\ref{defbarphi}) of $\bar\phi$ imply, in particular,
\begin{align}
\mf r&\ll \sqrt{c_1} \, g^{-\frac14} \label{rcondsfset}\\
\upepsilon&\ll \sqrt{c_1}g^{\frac14} \label{condeps2}\\
\hbar &\ll \sqrt{c_1}g^{\frac14} \label{condhbar}
\end{align}
and
\begin{align}
\mc L^{-2}g^{-\frac14}\ll\bar\phi\ll g^{-\frac14}\label{condbarphi}.
\end{align}
Remember that $\ell $ will be chosen large enough, depending on $\frac{\mc L}\ell $, and so we may regard $\mc L $ as a fixed constant in this proof.  We have
\begin{align}
\Vert g(X)\Vert &\leq \sup_{R\text{ s.t. }\overline R\subset X}\sup_{\substack{\upphi\in \mb R^\Lambda,\, D(2\Gamma(\upphi|_X) + \bar\phi)=R,\\ D(2\Gamma(\upphi|_{X^c}))\cap X=\emptyset \\\psi'\in\mc K_{\mf r}(\square_X) }}  e^{ - \int_{\xi\in \square_{X }}\la g\bar\phi^2\ra \cdot\psi_\xi^2 - \frac1{20}\Vert\upphi\Vert_X^2 } \\\nonumber &\qqquad\times e^{-\int_{\xi\in \square_{X\setminus R}} V_\eps(\psi_\xi+\psi'_\xi)  } |g'_R(X;\psi+\psi')|\Big|_{\psi = \mc A\Gamma \upphi},
\end{align}
where $V_\eps(x) $ is as in (\ref{defveps}) and
\begin{align}
g_R'(X;\psi) &= e^{\frac g4\bar\phi^4 |R| + \int_{\xi\in\square_{R}}(\hbar + \eps - \bar G^{-1}\bar\phi)\psi_\xi + \frac 32 g\bar\phi^2\psi_\xi^2 } g^R(X,\psi+\bar\phi)
\end{align}
Note that we have chosen the weight $ e^{-  \la g\bar\phi^2\ra\Vert\psi\Vert_{L^2( \square_ X)}^2 } $ in the norms so as to exactly compensate the basic stability bound (\ref{eqbasicstab}) on $V_\eps $:
\begin{align}
e^{-\int_{\xi\in \square_{X\setminus R}} V_\eps(\psi_\xi+\psi'_\xi) - \int_{\xi\in \square_{X }}\la g\bar\phi^2\ra \cdot\psi_\xi^2 }\leq e^{3|X\setminus R|\la \mf r(\upepsilon + 12g\bar\phi\,\mf r^2) \ra- \int_{\xi\in \square_{R }}\la g\bar\phi^2\ra \cdot\psi_\xi^2}
\end{align}
for any $\psi,\psi' $ of the above form. For such fields, by (\ref{lfbound}) and (\ref{sfcondshift}), (\ref{rcondsfset}), (\ref{condeps2}), (\ref{condhbar}), we also have
\begin{align}
| g'_R(X,\psi+\psi')|&\leq e^{  \la c_1\ra  |R|_\ell + \int_{\xi\in\square_{R}}  C g^{\frac14} |\psi_\xi+\psi'_\xi|  -\la g^\half\ra |\psi_\xi+\psi'_\xi|^2 - c_g\mc L (X) }\nnb &\leq e^{  \la c_1\ra   |R|_\ell -\frac 12  g^\half  \int_{\xi\in\square_{R}}  |\psi_\xi |^2 - c_g\mc L (X) }
\end{align}
with $ C\ll \sqrt{c_1} $. To compensate for the factor $\la c_1\ra |R|_\ell $, note that, by a simple modification of (10.48) in \gk, under the assumption (\ref{sfcondshift}), we have
\begin{align}
|R|_\ell &\leq \const c_2^{-2}\, g^\half \, \int_{\square_R} |\mc A \phi_\xi|^2
\end{align} 
for any $\phi $ with $\mc D(2\phi+ \bar\phi)=R$, with a constant that depends only on $\ell$. Therefore, if $ R=  D(2\Gamma(\upphi|_X)+\bar\phi) $,
\begin{align}
|R|_\ell &\leq \const c_2^{-2}\, g^\half \int_{\square_R} |[\mc A\Gamma(\upphi|_X)]_\xi |^2 \leq \const c_2^{-2}\, g^\half \,\mc L^4\, \Vert \mc A\Gamma\Vert_{\text{HS}}^2 \,\Vert\upphi\Vert_X^2
\end{align}
where
\begin{align}
\Vert \mc A\Gamma\Vert_{\text{HS}}^2 &= \sup_{x\in X}\sums{y\in X \\ z\in R} \left| \int_{\xi\in\square_z} (\mc A\Gamma)_{\xi,x}(\mc A\Gamma)_{\xi,y}  \right| 
\end{align}
We have the following, less subtle version of (\ref{sharpboundcov}) (see the Appendix): Under the assumption (\ref{condhsnorm}),
\begin{align}\label{roughboundcov2}
\Vert \mc A\Gamma\Vert_{\text{HS}}^2 &\leq \const (g\bar\phi^2)^{-1},
\end{align}
with a constant that depends only on $L$. Therefore, by (\ref{condbarphi}),
\begin{align}
|R|_\ell &\ll \const c_2^{-2}   \mc L^8 \Vert\upphi\Vert_X^2\ll \tfrac1{20 \la c_1\ra } \Vert\upphi\Vert_X^2
\end{align}
Since we assumed $c_2\geq c_2(c_1) $, the first assertion of the Lemma follows immediately.\\
We now turn to the bound on 
\begin{align}
\tilde g(\Delta;\uppsi;\upphi) &= g(\Delta;\uppsi;\upphi) - e^{-\int_{\xi\in\square_\Delta}V_\eps(\psi_\xi) }.
\end{align}
In the supremum over $\upphi\in\mb R^\Lambda $, we need to distinguish the cases $R(\upphi|_\Delta)=\emptyset $ and $R(\upphi|_\Delta)\neq \emptyset $. In the latter case, we have $g(\Delta;\uppsi;\upphi) = 0 $ (since $\overline {R(\upphi)}\subset \Delta$ cannot hold), and so we need to show that 
\begin{align}\label{1pboundlf}
\sup_{\substack{\upphi\in \mb R^\Lambda,\, D(2\Gamma(\upphi|_\Delta) + \bar\phi)\supset\Delta,\\ D(2\Gamma(\upphi|_{\Delta^c}))\cap \Delta=\emptyset \\\psi'\in\mc K_{\mf r}(\square_ \Delta) }}  e^{-\int_{\xi\in \square_{\Delta}} V_\eps(\psi_\xi+\psi'_\xi) + \la g\bar\phi^2\ra\psi_\xi^2 - \frac1{20}\Vert\upphi\Vert_\Delta^2  }\Big|_{\psi=\mc A\Gamma\upphi} \leq e^{-c(\ell) \, c_2^2}
\end{align}
This is follows in the same was as (\ref{gbound}). In the former case, if $R(\upphi|_\Delta) = \emptyset $, and $ \uppsi\in\Psi(\Delta,\upphi)$, we have 
\begin{align}
\tilde g(\Delta;\uppsi;\upphi) &= e^{-\int_{\xi\in\square_\Delta}V_\eps(\psi_\xi)} \big[g^\emptyset(\Delta;\psi+\psi'+\bar\phi)-1\big].
\end{align}
Therefore, using (\ref{blocksfbound}) and the basic stability bound (\ref{eqbasicstab}) as before,
\begin{align}
\sup_{\substack{\uppsi\in \Psi(\Delta,\upphi) \\\upphi\in \mb R^\Delta,R(\upphi)=\emptyset }} e^{-\la g\bar\phi^2\ra\int_{\xi\in\square_\Delta}\psi_\xi^2} |\tilde g(\Delta;\uppsi;\upphi)| \leq g^{\frac13}\, e^{3\ell^4\la \mf r(\upepsilon + 12g\bar\phi\,\mf r^2) \ra}\leq e^{-c(\ell)c_2^2}, 
\end{align}
and the second claim of the Lemma follows.

\end{proof}

\noindent
Note that the norms $\Vert g\Vert $, $\Vert g\Vert_\Delta $ and $\Vert \tilde g\Vert $ can be expressed in terms of $\Vert g(X)\Vert  $ of the Lemma as follows
\begin{align}
\Vert g\Vert &= \sup_{ \Delta\in\mc C  } \sums{X\supset\Delta\\ |X|_\ell>1} (2\,e)^{|X|_\ell-1} \Vert g(X)\Vert\\
\Vert g\Vert_\Delta  &= \sup_{ \Delta\in\mc C  } \Vert g(\Delta)\Vert\\
\Vert \tilde g\Vert &\leq  \max\Big\{ \max_{\Delta\in\mc C} \Vert \tilde g(\Delta)\Vert ,\; \max_{\substack{X\subset\square_\sfx \\ |X|_\ell>1}} \Vert g(X)\Vert^{1/|X|_\ell}  \Big\}
\end{align}
We now use the Lemma to derive bounds on the norms of $g$ and $\tilde g$.

\begin{prop}\label{boundsinput}

Assume (\ref{stabcond}), (\ref{sfcondshift}) and (\ref{1blockcond2}). Then
\begin{align}
\Vert g\Vert_\Delta^{(\frac{\mc L}\ell)^4} &\leq \la1 \ra\label{1blockbound2}\\
\Vert \tilde g\Vert &\leq e^{-\frac {c_g}3 \ell}\label{boundtildeginput}\\
\Vert g\Vert &\leq \sum_{n=2}^{|\Lambda|_\ell} \tfrac{ C(\ell)^{n-1} }{n-1}\label{boundinputg}
\end{align}
where
\begin{align}
C(\ell) &= 30\sum_{0\neq k\in\mb Z^4}e^{-c_g\,\ell\, |k|}
\end{align}

\end{prop}

\begin{proof}

The bound on $\Vert g\Vert_\Delta$ follows immediately from (\ref{1blockcond2}) and Lemma \ref{stabboundact}. For the second bound, using that $\mc L(X) \geq \ell \,(|X|_\ell-1)  $ we get
\begin{align}
\Vert g(X)\Vert^{1/|X|_\ell} &\leq \la  e^{ -c_g\,\ell\, \frac{|X|_\ell-1}{|X_\ell|} }\ra \leq e^{-\frac {c_g}3 \ell} 
\end{align}
if $|X|_\ell>1 $, and since $ \Vert \tilde g(\Delta)\Vert \leq e^{-c(\ell)c_2^2} $ can be chosen much smaller than this for $c_2 \geq c_2 (c_1(\ell)) $ (see Theorem \ref{thmgk}), the second bound follows.\\
Finally, for the third bound, we note that Lemma \ref{stabboundact} implies
\begin{align}
\Vert g\Vert &\leq \sup_{ \Delta\in\mc C  } \sums{X\supset\Delta\\ |X|_\ell>1} \la 2\, e\ra^{|X|_\ell-1} e^{-c_g\mc L(X)}\nnb
&= \sum_{n=2}^{|\Lambda|_\ell}\tfrac{\la 2\, e\ra^{n-1}}{(n-1)!} \sup_{\Delta_1\in\mc C}\sums{\Delta_2,\ldots,\Delta_n\in\mc C\\ \Delta_m\text{ all different}} \max_{T\text{ tree on }\ul n} \prod_{\{m,m'\}\in T} e^{-c_g |c(\Delta_m)-c(\Delta_{m'})| },
\end{align}
where $c(\Delta) $ is the center of $\Delta $. Bounding the maximum over trees by a sum, inductively removing vertices of the tree (similar, but much easier, than the argument in the proofs of Proposition \ref{propboundme} or Lemma \ref{pinchandsum}), and performing he sum over trees as in (\ref{treebound}), we get 
\begin{align}
\Vert g\Vert &\leq \sum_{n=2}^{|\Lambda|_\ell} \tfrac{  C(\ell)  ^{n-1} }{n-1}
\end{align}
with 
\begin{align}
C(\ell) &=  \sup_\Delta\sum_{\Delta'\neq\Delta} e^{-c_g|c(\Delta) - c(\Delta')|}\leq 30\,\sum_{0\neq k\in\mb Z^4}e^{-c_g\,\ell\, |k|}
\end{align}
This proves the Lemma. 

\end{proof}

\para{The choice of parameters.}{pchoicemthm} The construction of $\log Z $ in Proposition \ref{propce}, and the norms used in section \ref{bounds} to control it, depend on parameters $\mc L $ as well as $\upepsilon$ and $\mf r  $. We shall now describe the conditions on these parameters that are needed to ensure that the construction produces an error $M'(h) $ (cf. (\ref{defmp})) to the magnetization that is smaller than the main term $L^{-n}\bar\phi $. We will then make a choice for the parameters which satisfies all conditions, as long as the scale $\ell $ of the cluster expansion of Theorem \ref{thmgk} is large enough, and the magnetization scale $n$ at which we apply Theorem \ref{thmgk} is chosen appropriately. (Obviously the initial coupling constant $g_0 $ also has to be small enough for Theorem \ref{thmgk} to hold). Increasing $\mc L,\ell $ improves the size of the error $M'(h) $ (at the expense of stronger smallness conditions on $g_0,h$), and we shall put all pieces together in the proof of Theorem \ref{mainthm} and show that the relative error can be made smaller than any $\delta>0 $.\\
The bound (\ref{boundme}) on the Mayer expansion, together with the Cauchy formula, imply that $M'(h)\ll L^{-n}\bar\phi $ if 
\begin{align}
\sup_{\substack{\sfx\in\Uplambda \\  |\eps|\leq\upepsilon}} |\sfA(\{\sfx\})-1|&\ll1\label{cond1st}\\
\Vert \sfA\Vert&\ll 1 \label{cond2nd}
\end{align}
and $\upepsilon\gg (\mc L^4\bar\phi)^{-1} $, which is the same condition as (\ref{condeps}). These bounds could be established by our machinery if
\begin{enumerate}
	\item All assumptions that were used in the proofs of the Propositions of section \ref{bounds} hold. More precisely, we assumed (\ref{condhsnorm}) in the proof of Proposition \ref{proplocbound} (which contributes to the bound on $\sfA$), and (\ref{1blockcond1}) in the proof of Proposition \ref{bound1p} (which contributes to the bound on $|\sfA(\{\sfx\})-1| $).
	\item All assumptions that were used in the proof of Proposition \ref{boundsinput} hold. More precisely, we assumed the basic small field condition (\ref{sfcondshift}) for $\bar\phi$, the basic stability condition (\ref{stabcond}), and the conditions (\ref{1blockcond2}), (\ref{condconvce}).  
	\item The series (\ref{boundme}), (\ref{boundloc}), (\ref{boundrebl}), (\ref{boundinputg}) that are the output of these Propositions are bounded uniformly in the volume. These estimates contribute to our final bound on $\Vert\sfA\Vert$. More precisely, 
	\begin{enumerate}
		\item Whenever $\ell $ is large enough, (\ref{boundinputg}) proves that $\Vert g\Vert $ is small, depending on $\ell$.
		\item Whenever $(\frac{\mc L}\ell)^4\Vert g\Vert $ is small enough, (\ref{boundrebl}) proves that $\Vert \sfg\Vert $ is small, depending on $(\frac{\mc L}\ell)^4\Vert g\Vert $, and that $\Vert\sfg\Vert_\square $ is order $1$.
		\item Whenever (\ref{condconvce}) is satisfied, (\ref{boundloc}) and (\ref{boundis}) prove that $\Vert  \sfA\Vert $ is small, depending on $\Vert\sfg\Vert $ and $(\mc L^2\,\mf r\, g^\half\bar\phi)^{-1}\ll1  $ (see (\ref{basiccovestimate})).
	\end{enumerate}
	For step (b), we also needed (\ref{singleblockbounda}), which follows from (\ref{1blockbound2}) of Proposition \ref{boundsinput}.
	\item The series in (\ref{ebound1p}) is bounded uniformly, and the estimate (\ref{bound1pact}) on $\sfB(\square_\sfx) $ holds. These bounds contribute to our final estimate on $|\sfA(\{\sfx\})-1| $. More precisely, 
	\begin{enumerate}
		\item If $(\frac{\mc L}\ell)^4\Vert\tilde g\Vert $ is small, then so is $|\sfA(\{\sfx\}) - \sfB(\square_\sfx)| $, by (\ref{ebound1p}). Note that, by (\ref{boundtildeginput}), this is the case whenever $\ell $ is large enough, depending on $\frac{\mc L}\ell $.
		\item If the only condition of section \ref{sbheur} that hasn't been used already, namely (\ref{1blockcond3}), holds, then, by (\ref{bound1pact}), $|\sfB(\square_\sfx)-1| $ is small.
	\end{enumerate}
\end{enumerate}
There is an important difference\footnote{Another relevant difference is that, by 3(c), the size of $\Vert\sfA\Vert$ does not only depend on $(\mc L^2\,\mf r\, g^\half\bar\phi)^{-1} $ (which we were able to make $g$ dependent in the heuristics), but also on $\ell$, which is independent of $g$.} between the discussion above and the analogous discussion for the conditions on parameters in the heuristic section \ref{cparaheur}. Namely, by points 3(a) and 3(b), and also by point 4(a), $\ell $ has to be chosen large, depending on $\frac{\mc L}\ell $, and since the smallness condition on $g$ in Theorem \ref{thmgk} depends on $\ell$, in particular, $\mc L $ cannot be chosen $g$ dependent (in contrast to (\ref{regl})). This puts stronger limits the possible choices of $\upepsilon,\mf r $ and the magnetization scale $n$ than what was discussed in section \ref{cparaheur}. Most prominently, the conditions $\bar\phi\ll g^{-\frac14}\ll \mc L\bar\phi $ (see (\ref{sfcondshift}) and (\ref{lfcondshift})) now restrict our choice of the magnetization scale to 
\begin{align}\label{regnrigorous}
n = \tfrac 13\log_L h^{-1} -(\tfrac13-\tfrac14) \log_L\log h^{-1} - x& \qquad\text{with}\qquad  1\ll x\ll\log_L\mc L     ,
\end{align}
which is to be compared to (\ref{regn}).\\[5pt]
Despite the limited freedom to choose the parameters, there is sufficient flexibility to obtain a bound on $M'(h) $ that is smaller than any $\delta>0$ relative to the leading contribution $L^{-n}\bar\phi $ to the magnetization. We shall show in the next paragraph that, for given $\delta$, the following choice proves Theorem \ref{mainthm}:
\begin{align}
n &= \big\lfloor\tfrac 13\log_L h^{-1} -(\tfrac13-\tfrac14) \log_L\log h^{-1} +\log_L \delta\big\rfloor\nnb
\mf r &= \delta^2 \bar\phi\nnb
\mc L &= L^{\lfloor\log_L\delta^{-\frac 74}\rfloor} \label{finalchoice}\\
\upepsilon &=  (\mc L^4\,\bar\phi\,\delta)^{-1}\nnb
\ell &= \text{the smallest fraction of }\mc L \text{ such that }\ell e^{\frac{c_g}{12}\ell}\geq \mc L \nonumber
\end{align}
Note that, with our choice of $n$, $\bar\phi\leq \const \delta\, g^{-\frac14} $, with a constant that depends only on $L$. It is easy to check the above choices fulfill all conditions listed in points $1-4$ above, if $\delta$ is small enough, and $g_0 $ and $h $ are small enough depending on $\delta$.

\begin{rem}\label{remgkimprovement}
The parameter $\ell$ is used in Theorem \ref{thmgk} to encode the smallness of activities $\bar g(X)$ on large polymers $X$ (which essentially comes from perturbative arguments), as well as the size of the small field region (\ref{defsfs}), via the constant $c_2 = c_2(\ell)$. In Balaban's approach to the renormalization group (see, e.g. \cite{Balaban1983,doi:10.1142/S0129055X13300100,Balaban2010}), it is, in fact, shown that activities decay as $g^{\epsilon |\sfX|} $, and analyticity can be controlled in small field regions of the size $g^{-\frac14 - \epsilon} $, for some small $\epsilon $ (say, $\epsilon = \frac1{100} $). That is, it may be conjectured that the arguments of this paper also hold for a choice $\ell = g^{-\epsilon} $, if $\epsilon $ is small enough (but independent of the initial coupling $g_0$). If this was the case, we could also choose $\delta = g^{\epsilon} $, and Theorem \ref{mainthm} would, in fact, isolate exactly the leading term in the critical asymptotics of the magnetization (cf. Remark \ref{remstrongsympt}). Unfortunately (or, in the grand scheme of things, perhaps fortunately), the critical point of $\phi_4^4 $ has not been constructed by Balaban's method. We therefore restrict ourselves to finding the leading critical behavior of $m(h)$ only within an arbitrarily small error $\delta$.

\erem

\end{rem}

\para{Proof of the main theorem.}{ppmt}
Fix $L$, assume that $\delta >0 $ is small, that $\ell $ is chosen as in (\ref{finalchoice}), and that $g_0$ is small, depending on $\ell$, so that the conditions of Theorem \ref{thmgk} are satisfied. For any small enough $h>0$, define $n$ as in (\ref{finalchoice}), and apply Theorem \ref{thmgk} at this scale. It is easy to see, using (\ref{eqpert}), that the minimizer $\bar\phi $ of (\ref{defbarphi}) is unique, and that it satisfies $\bar\phi \leq\const  \delta\,g^{-\frac14} $.\\ 
For this choice of $\bar\phi$, we have that the finite volume magnetization $M(h) = L^{-n}\bar\phi + M'(h) $, where $M'(h) $ is given in (\ref{defmp}). Further, for $\mc L $ chosen as in (\ref{finalchoice}), we have $M'(h) = \frac1{L^n|\Lambda|}\partial_{\eps=0}\log Z(\eps)  $, where $\log Z $ is given by the Mayer expansion (\ref{mealg}), as long as this expansion converges.\\
We will presently check that, indeed, the Mayer expansion for $\log Z(\eps) $ converges for all $|\eps|\leq\upepsilon $, with $\upepsilon $ as in (\ref{finalchoice}), and that $|\log Z(\eps)|\leq \frac12 |\Uplambda| = \frac 12\mc L^{-4}|\Lambda| $ in this region. Using the Cauchy formula, this implies that $|M'(h)|\leq\frac 12 L^{-n} \mc L^{-4}\upepsilon^{-1} =  L^{-n}\bar\phi\cdot \frac\delta2 $, and thus 
\begin{align}
M(h) = L^{-n}\bar\phi \,\big(1+\mc E'(h)\big)
\end{align}
with $|\mc E'(h)|\leq \frac \delta2 $. Since $$L^{-n}\bar\phi = \Big(\frac{3h\log h^{-1}} {16\pi^2}\Big)^{\frac13} \big(1+O(1/\log\log h^{-1})\big) $$ by (\ref{hovergheur}), we conclude
\begin{align}\label{logcorrfv}
M(h) = \Big(\frac{3h\log h^{-1}} {16\pi^2}\Big)^{\frac13}(1+\mc E(h)), 
\end{align}
which is the finite volume version of the claim (\ref{logcorr2}) of the Theorem.\\
To prove convergence of the expansion and $|\log Z |\leq \frac12 |\Uplambda| $, note the following. By Proposition \ref{boundsinput}, $\Vert g\Vert_\Delta^{(\mc L/\ell)^4} \leq \frac 32 $, and $\Vert \tilde g\Vert,\Vert g\Vert\leq o(1)$ (we write $o(1)$ to mean ``as small as we like, if $ \delta>0$ is small enough''). Therefore, Propositions \ref{bound1p} and \ref{boundreblock} are applicable, and, by (\ref{finalchoice}), we get 
\begin{align}
|\sfA(\{\sfx\}) - \sfB(\square_\sfx)| &\leq o(1)\nnb
\Vert \sfg\Vert_\square&\leq \tfrac53 \\
\Vert \sfg\Vert &\leq o(1).\nonumber
\end{align}
By the first bound, and (\ref{bound1pact}), we get $|\sfA(\{\sfx\}) -1 | = o(1) $, and therefore also $\log \sfA(\{\sfx\})= o(1) $, $|\sfA(\{\sfx\})|^{-1} = 1+o(1) $. By the other two bounds, Proposition \ref{proplocbound} is applicable, and we get, by (\ref{finalchoice})
\begin{align}
\Vert A\Vert \leq o(1).
\end{align}
By Proposition \ref{boundis}, also $\Vert \sfA\Vert\leq o(1) $. Thus, Proposition \ref{propboundme} can be applied, and the Mayer expansion converges and $|\log Z|\leq |\Uplambda|\cdot o(1) $. All this holds uniformly in $|\eps|\leq\upepsilon. $ As soon as $\delta>0$ is small enough, we conclude the claim $|\log Z|\leq \frac 12|\Uplambda|  $ and with it (\ref{logcorrfv}).\\
To finish the proof of the Theorem, we need to show that (\ref{logcorrfv}) carries over to the thermodynamic limit. We give a sketch of the argument, similar to the sketch discussed in section 12 of \gk. Indeed, the discussion given there implies that: The coupling constants $g,\nu,z $ of (\ref{eqpert}) have an infinite volume limit; The kernels $\bar G,\mc A $ are the periodizations of infinite volume kernels satisfying the same estimates; and the activities $g^D $ converge in the following sense. Identify $\Lambda_0 $ with a block of size $L^{4N}  $ in $\mb Z^4 $, so that sets with different values of $N$ are nested consistently. For finite sets $X,D\subset\mb Z^4 $, define the domains $\mc K(X) $, $\mc D(D,X) $ as before, but with fields $\phi\in\mb C^{\mb Z^4} $ (resp. $\mb R^{\mb Z^4} $) with finite support, and with the infinite volume kernel $\mc A$ instead of the periodized one. For fixed, finite $X\subset\mb Z^4$, if $N$ is large enough so that $X$ can be identified with a subset of $\Lambda_0$, these domains are included in the finite volume domains (modulo a change of constants that becomes negligible as $N\to \infty$), and thus the activities $g^D(X) $ define analytic functions on these domains. \gk argue inductively that these analytic functions converge almost uniformly as $N\to\infty $.\\
Our activities $A(\sfX) $ can be expressed in terms of $g^D(X) $, $D\subset X\subset \square_\sfX $, through finite sums and derivatives, and therefore, they, too, define analytic functions on the infinite volume domains which converge almost uniformly. By dominated convergence, therefore also the numbers $\sfA(\sfX) $ converge in the thermodynamic limit, for any finite $\sfX$. Since the Mayer expansion converges absolutely, uniformly in $N$, $M'(h) = M'(h,N) $ converges as $N\to\infty $. Since the minimizer $\bar\phi $ defined in (\ref{defbarphi}) obviously also converges when $g,\nu$ do (and $n$ is independent of $N$), the finite volume magnetization $M(h) = M(h,N) $ converges, and the limit $m(h)$ satisfies the bound (\ref{logcorr2}). This proves the Theorem.\qed

\appendix
\section{Estimates for the covariance} 

In this Appendix, we prove the bounds (\ref{roughboundcov}), (\ref{sharpboundcov}), (\ref{basiccovestimate}), and (\ref{roughboundcov2}). The operators appearing in these bound implicitly depend on the volume $L^{4N}$, but, as noted in Remark \ref{remivl}, they are simply the periodizations of the $N=\infty$ operators. It is easy to see that the finite $N$ inequalities (\ref{roughboundcov}), (\ref{sharpboundcov}), (\ref{basiccovestimate}) are a consequence of the same inequalities at $N=\infty$, whenever $N$ is large enough (possibly depending on $h$). We will thus focus our attention on the case $N=\infty$ (which has the slight advantage that we can use the usual Fourier transform, as opposed to the discrete one).\\
Our bounds are based on the following Lemma, in which we abbreviate $m^2=3 g\bar\phi^2+\nu$, where $\nu$ is as in Theorem \ref{thmgk}. Note that, since we assume $1\ll \bar\phi\ll g^{-\frac14} $, we have $m^2 = \la 3 g \bar\phi^2\ra $.

\begin{lem}
(i) The kernel $\Gamma $ can be written as
\begin{align}\label{decompgamma}
\Gamma_{x,y} &= \sum_{j=0}^{\lfloor\log m^{-1}\rfloor} e^{-3j} \,\gamma_j\big(e^{-j}(x-y)\big)
\end{align}
where $\gamma_j(a),a\in\mb R^4 $ is even and satisfies the bounds $\Vert \gamma_j\Vert_\infty\leq C $ and
\begin{align}
\sum_{a\in e^{-j}\mb Z^4}|\gamma_j(a)|&\leq C\, e^{4j}\\
\sum_{a\in e^{-j}\mb Z^4}|\gamma_j(a+c) - \gamma_j(a)|&\leq C\, e^{3j} & \forall |c|_1&=e^{-j}\\
\sum_{a\in e^{-j}\mb Z^4} |\gamma_j(a+c+c')-\gamma_j(a+c')-\gamma_j(a+c)-\gamma_j(a)|&\leq C\, e^{2j} & \forall |c|_1&=|c'|_1=e^{-j}.
\end{align}
The constant $C$ is uniform in $m^2 $. \\
(ii) let $\tilde\Gamma = \Gamma - (-z\, \Delta_{\mb Z^4} + m^2)^{-\half} $. Then 
\begin{align}
\sup_{x\in\mb Z^4}\sum_{y\in\mb Z^4}|\Gamma_{x,y}| &\leq C
\end{align}
\end{lem}

\begin{proof}

We have
\begin{align}
\Gamma_{x,y} &= \tfrac1{(2\pi)^4} \int_{|p|_\infty\leq\pi} \hat\Gamma(p) e^{ip(x-y)}\ud p,
\end{align}
where
\begin{align}
\hat \Gamma(p) &= \Big(\bar\mu(p) + \tfrac32 g\bar\phi^2 \widehat{\mc A^\star\mc A}(p)\Big)^{-\half},
\end{align}
$\bar\mu(p) $ is as in Theorem \ref{thmgk}, and $\widehat{\mc A^\star\mc A}$ is the Fourier transform of $\mc A^\star\mc A $, which, by (\ref{boundmca}) and (\ref{mcai1}), (\ref{mcai2}), is a $2\pi $ periodic, nonnegative function, analytic for $|\Im p|\leq c'_{\mc A} $ with $c'_{\mc A}>0 $ depending only on $L$, and with $\widehat{\mc A^\star\mc A}(0)=1 $. By the properties of $\bar\mu(p) $ and these observations on $\widehat{\mc A^\star\mc A} $, there exist constants $c,c'>0 $ and $C'<\frac1{2c} $, depending only on $L$, such that $\hat \Gamma(p)\leq c' $ if $|p|_2\geq c $ and 
\begin{align}
|\hat\Gamma(p)^{-2} - zp^2 - m^2|\leq C'(m^2|p|_2 + |p|_2^3)\qquad\text{for}\qquad |p|_2< e\cdot c
\end{align}
For some $j=0,\ldots,J:= \lfloor\log m^{-1}\rfloor  $, consider now the function $\hat\Gamma_j'(q) = e^{-j}\hat\Gamma(e^{-j}q) $, restricted to the annulus $e^{-1}\leq|q|_2\leq e $. If $e^{-j-1}\geq c $, we obviously have $\hat\Gamma'_j(q)\leq c' $, and also all derivatives of $\hat\Gamma'_j $ are bounded independently of $j$. If $e^{-j-1}<c $, we have
\begin{align}
|\hat\Gamma_j'(q)^{-2} - z\, q^2 + e^{2j}m^2|\leq C'\big( e^{j}\,m^2,|q|_2 + e^{-j}|q|_2^3   \big)
\end{align} 
and thus $|\hat \Gamma'_j(q)| \leq $ const, for a $j$ independent constant. Also, all derivatives of $\hat\Gamma_j' $ are bounded by $j$ independent constants. If $j=J $, since $e^{2J}m^2 = O(1) $, this even holds for $|q|_2\leq e $. \\
It is standard to conclude from these properties of $\hat\Gamma $ claim (i) of the Lemma. Indeed, let $\chi_j $, $j=0,\ldots,J  $, be a smooth partition of unity on $\{|p|_\infty\leq\pi\} $ such that supp $\chi_j\subset \{e^{-j-1}\leq |p|_2\leq e^{-j+1} \} $ for $j<J $ and supp $\chi_J\subset \{|p|_2\leq e^{-J+1}\} $, and such that $\chi_j'(q):= \chi_j(e^j\, p) $ and all its derivatives are bounded by $j$ independent constants. Defining
\begin{align}
\gamma_j(a) &= \tfrac1{(2\pi)^4}\int_{\mb R^4} \hat\Gamma_j'(q)\, \chi_j'(q)\, e^{iqa} \ud q,
\end{align}
we clearly have (\ref{decompgamma}). By the above discussion, this is the Fourier transform of a smooth, compactly supported function, which is bounded along with its derivatives, uniformly in $j$. The bounds on $\gamma_j$ can be deduced from this fact by the standard Fourier correspondence between smoothness and decay.\\
For claim (ii) of the Lemma, note that
\begin{align}
\tilde\Gamma_{x,y} &= \tfrac1{(2\pi)^4} \int_{|p|_\infty\leq\pi} \hat{\tilde\Gamma}(p) e^{ip(x-y)}\ud p,
\end{align}
with 
\begin{align}
\hat{\tilde\Gamma}(p)= \hat\Gamma(p) - \hat \Delta(p) ^{-\frac12} = \frac{\hat\Gamma(p)^{-2} - \hat\Delta(p)}{\hat\Delta(p) \hat\Gamma(p)^{-1} + \hat\Gamma(p)^{-2}\hat\Delta(p)^\half } ,
\end{align}
where
\begin{align}
\hat\Delta(p) &=   z \, {\textstyle\sum_{d=1}^4} [2-2\cos(p_d)] + m^2  .
\end{align}
It is easy to check that $\hat{\tilde\Gamma}(p) $ is bounded uniformly in $m^2$, together with all its derivatives. Standard Fourier analysis then implies (ii).

\end{proof}

\noindent
We now use the Lemma to prove (\ref{roughboundcov}), (\ref{sharpboundcov}), (\ref{basiccovestimate}), and (\ref{roughboundcov2}). Since, by (\ref{decompgamma}), $\Gamma_{x,y} $ is bounded, (\ref{roughboundcov}) follows easily. For (\ref{basiccovestimate}), define 
\begin{align}
\Vert\Gamma\Vert_{j,k} &= \max\Big\{\sup_{\sfx\in\Uplambda}\sum_{\sfx'\in\Uplambda}\mc L^{4\delta_{j>1} + 2\delta_{k>1} - 2\delta_{k=1}} |\ddot\Gamma_{\sfx',j;\sfx,k}| ,\sup_{\sfx\in\Uplambda}\sum_{\sfx'\in\Uplambda}\mc L^{4\delta_{j>1} + 2\delta_{k>1} - 2\delta_{k=1}} |\ddot\Gamma_{\sfx,j;\sfx',k}| \Big\}
\end{align}
so that 
\begin{align}
\Vert\Gamma\Vert &= \sup_{k\in\ul{\mc L^4}}\sum_{j\in\ul{\mc L^4}}\Vert\Gamma\Vert_{j,k}.
\end{align}
It is easy to see that the Lemma implies the rough bounds
\begin{align}
\Vert\Gamma\Vert_{j,k}&\leq \const  \left\{\begin{array}{ll} \mc L^{-2} m^{-1}   & j=k=1 \\ \mc L^2\log m^{-1}  &  j=1,k>1 \\  \mc L^6 \log m^{-1} & k=1,j>1  \\  \mc L^{10}  & j,k>1    \end{array}\right.
\end{align}
and (\ref{basiccovestimate}) follows immediately from this if $\mc L^{12}m\log m^{-1} \ll1 $, which is (\ref{condhsnorm}).\\
Finally, (\ref{sharpboundcov}) is a sharper version of (\ref{roughboundcov2}), and so we will only prove the former. Set 
\begin{align}
\Vert\mc A\Gamma\Vert^2_{k,k' } &= \sup_{\sfx\in\Uplambda}\sum_{\sfx'\in\Uplambda} \mc L^{-2\delta_{k=1} - 2\delta_{k'=1} + 2\delta_{k>1} + 2\delta_{k'>1}}\sum_{\sfy\in\Uplambda} \left|\int_{\xi\in \boxdot_\sfy}\dot{{\mc A}\Gamma}_{\xi;\sfx,k}\dot{{\mc A}\Gamma}_{\xi;\sfx',k'}  \right|
\end{align}
so that 
\begin{align}
\Vert\mc A\Gamma\Vert_{\text{HS}(\sfX,\sfX')}^2\leq \sup_{k\in\ul{\mc L^4}}\sum_{k'\in\ul{\mc L^4}} \Vert\mc A\Gamma\Vert^2_{k,k' } 
\end{align}
for all $\sfX,\sfX'\subset\Uplambda $. In the same way as before, it is easy to see that
\begin{align}
\Vert\mc A\Gamma\Vert^2_{k,k' } \leq \const  \left\{\begin{array}{ll} \mc L^4  m^{-1} \log m^{-1}  &  k=1,k'>1 \text{ or }k'=1,k>1  \\  \mc L^{8}\log^2m  & k,k'>1    \end{array}\right.
\end{align}
and so (\ref{sharpboundcov}) will follow if we can prove $ \Vert\mc A\Gamma\Vert^2_{1,1 } \leq \la m^{-2}\ra $. To see this, note that, for $\xi\in\square_y $, by (\ref{mcai1})
\begin{align}
\dot{{\mc A}\Gamma}_{\xi;\sfx,1} &= \sum_{x\in\square_\sfx}  \Big[\Gamma_{y,x} + \sum_{z\in\Lambda} \mc A_{\xi,z}  (\Gamma_{z,x} - \Gamma_{y,x})  \Big]
\end{align}
It is easy to use the exponential decay of $\mc A $ and the Lemma to conclude from this that
\begin{align}
\Vert\mc A\Gamma\Vert^2_{1,1 } \leq \Vert \Gamma\Vert_{1,1}^2 + \const \mc L^4 m^{-1}\log m^{-1}, 
\end{align}
where
\begin{align}
\Vert \Gamma\Vert_{1,1}^2 &= \mc L^{-4} \sup_{\sfx\in\Uplambda}\sum_{\sfx'\in\Uplambda}  \sum_{\sfy\in\Uplambda} \bigg|\sums{y\in \square_\sfy\\x\in\square_\sfx\\x'\in\square_{\sfx'}}\Gamma_{y,x }\Gamma_{y,x'}  \bigg|
\end{align}
By part (ii) of the Lemma, we have
\begin{align}
\mc L^4 \Vert \Gamma\Vert_{1,1}^2 &\leq \Vert (-z\,\Delta + m^2)^{-\frac12}\Vert^2_{1,1}  + 2\Vert \Gamma\Vert_{1,\infty}\Vert \tilde\Gamma\Vert_{1,\infty} + \Vert \tilde\Gamma\Vert_{1,\infty} ^2\nnb
&\leq \Vert (-z\,\Delta + m^2)^{-\frac12}\Vert^2_{1,1} + \const m^{-1},
\end{align}
and since 
\begin{align}
(-z\,\Delta + m^2)_{x,y}^{-\frac12} &=  \frac1\pi \int_0^\infty (-z\,\Delta+m^2+s)_{x,y}^{-1}s^{-\frac12}\ud s,
\end{align}
is non-negative just like the resolvent of the Laplacian itself, we have $  \Vert (-z\,\Delta + m^2)^{-\frac12}\Vert^2_{1,1} = m^{-2},$ and (\ref{sharpboundcov}) follows.

\bibliographystyle{alpha}
\bibliography{magnetization.bib}

\end{document}